\def\draft{0}

\documentclass[letterpaper,11pt]{article}
\usepackage{graphicx} 
\usepackage{style}

\title{Expanders Meet Reed-Muller: Easy Instances of Noisy k-XOR}

\author{Jaros{\l}aw B{\l}asiok\thanks{Bocconi University. \texttt{jaroslaw.blasiok@unibocconi.it}.} \and Paul Lou\thanks{Bocconi University, BIDSA. \texttt{paul.lou@unibocconi.it}.} \and Alon Rosen\thanks{Bocconi University, BIDSA. \texttt{alon.rosen@unibocconi.it}.}\and Madhu Sudan\thanks{Harvard University. \texttt{madhu@cs.harvard.edu}.}}

\begin{document}

\maketitle

\begin{abstract}
In the noisy $k$-XOR problem, one is given $y \in \bF_2^\constraints$ and must distinguish between $y$ uniform and $y = A x + e$, where $A$ is the adjacency matrix of a $k$-left-regular bipartite graph with $\variables$ variables and $\constraints$ constraints, $x\in \bF_2^\variables$ is random, and $e$ is noise with rate $\eta$. Lower bounds in restricted computational models such as Sum-of-Squares and low-degree polynomials are closely tied to the expansion of $A$, leading to conjectures that expansion implies hardness. We show that such conjectures are false by constructing an explicit family of graphs with near-optimal expansion for which noisy $k$-XOR is solvable in polynomial time. 

Our construction combines two powerful directions of work in pseudorandomness and coding theory that have not been previously put together. Specifically, our graphs are based on the lossless expanders  of Guruswami, Umans and Vadhan (JACM 2009). Our key insight is that by an appropriate interpretation of the vertices of their graphs, the noisy XOR problem turns into the problem of decoding Reed-Muller codes from random errors. Then we build on a powerful body of work from the 2010s correcting from large amounts of random errors. Putting these together yields our construction. 

Concretely, we obtain explicit families for which noisy $k$-XOR is polynomial-time solvable at constant noise rate $\eta = 1/3$ for graphs with $\constraints = 2^{O(\log^2 \variables)}$, $k = (\log \variables)^{O(1)}$, and $(\variables^{1-\alpha}, 1-o(1))$-expansion. Under standard conjectures on Reed-Muller codes over the binary erasure channel, this extends to families with $\constraints = \variables^{O(1)}$, $k=(\log \variables)^{O(1)}$, expansion $(\variables^{1-\alpha}, 1-o(1))$ and polynomial-time algorithms at noise rate $\eta = \variables^{-c}$.
\end{abstract}

\mnote{Set flag draft to 0 before submitting. Don't remove this note.}

\section{Introduction}
\label{sec:intro}

Noisy $k$-XOR is a canonical hypothesis-testing problem for sparse random linear systems and a central example in the study of statistical-computational gaps. A widely held intuition attributes such gaps to expansion properties of an associated graph. We show that this intuition fails in general by providing an explicit counterexample.

This serves as a cautionary tale: near-optimal expansion does not, by itself, imply hardness for noisy $k$-XOR. While many prior works focus on the case of constant $k$, for example $k=3$, our results apply when $k$ is polylogarithmic in the number of variables. We are not aware of any compelling reason to expect this regime to be algorithmically easier than the constant-$k$ setting.

\subsection{\texorpdfstring{The Noisy $k$-XOR Problem}{The Noisy k-XOR Problem}}

Given a vector $y \in \bF_2^{\constraints}$, one must distinguish between $y$ being sampled from the null distribution where $y$ is uniform random and the planted distribution where $y = A_H\cdot x + e$ that we now describe.
In the planted distribution, we consider matrices $A_H \in \bF_2^{\constraints \times \variables}$ where every row is of Hamming weight $k$, a uniform random $x \in \bF_2^{\variables}$, and a sparse noise vector $e \in \bF_2^\constraints$.
Viewing $A_H$ as the adjacency matrix of a $k$-left regular bipartite graph $H$ where each of the $\constraints$ rows corresponds to a constraint and each of the $\variables$ columns corresponds to a variable, we have the following definition.

\begin{definition}[Noisy k-XOR Distinguishing Problem]
\label{def:k-xor-problem}
    In the $\eta$-noisy XOR distinguishing problem associated with a constraint graph $H$, we are given a vector $y \in \bF_2^\constraints$, with a promise that either 
    \begin{itemize}
        \item (Null case): $y \sim \Unif(\bF_2^\constraints)$,
        \item (Planted case): $y = A_H \cdot x + e$, where $x \in \bF_2^\variables$ is a uniform random vector, and each element of the noise vector $e \in \bF_2^\constraints$ is independently $1$ with probability $\eta$ (and $0$ otherwise).
    \end{itemize}
    The goal of the algorithm is to distinguish between those two distributions. For example, an algorithm $\mathcal{A}$ is said to succeed in the distinguishing task, if 
    \begin{equation*}
        \left|\Pr_{y\sim \Unif{\bF_2^\constraints}}[\mathcal{A}(y) = 1] - \Pr_{\substack{x\sim \bF_2^\variables \\ e \sim \mathsf{Ber}(\eta)^{\constraints}}}[\mathcal{A}(A_H \cdot x  + e) = 1]\right| > 1/10.
    \end{equation*}
\end{definition}
When $\constraints \leq \variables$, the planted and the null distributions are typically identical (as long as the adjacency matrix $A_H$ is of full rank), so the problem is not well-posed. 
We therefore only consider the setting where $\constraints > \variables$. 
In this setting, the columns of the matrix $A_H$ span a linear code in $\mathbb{F}_2^\constraints$, and the problem is essentially equivalent to deciding whether $y$ is a randomly corrupted codeword (through a binary symmetric channel with noise level $\eta$), or a random word from $\bF_2^\constraints$.

\subsection{Statistical-Computational Gap and Expansion Properties}

When $H$ is a random \emph{dense} graph, or equivalently when $A_H$ is a random linear code, the problem is the classical Learning Parity with Noise (LPN) problem, which is conjectured to remain subexponentially hard even with subexponentially many samples~\cite{DBLP:conf/crypto/BlumFKL93,FOCS:Alekhnovich03}.

In the \emph{sparse} regime, a large body of work has studied random $k$-left-regular constraint graphs $H$. This includes hardness and cryptographic formulations of sparse noisy parity/$k$-LIN/XOR-code problems~\cite{FOCS:Alekhnovich03,bogdanov-sabin-vasudevan,DBLP:conf/crypto/BogdanovRT25}, as well as algorithmic work on planted and random CSPs~\cite{NIPS:FPV15,arXiv:BHLM25}. The sparse regime is believed to exhibit a statistical-computational gap: the problem becomes polynomial-time solvable beyond a known computational threshold, but is conjectured to remain hard below it. Some of this line of work was motivated by Feige's reduction from refuting random $k$-SAT to refuting random $k$-XOR~\cite{feige2002relations,fko06}.

Closely related, though somewhat orthogonal to the noisy $k$-XOR distinguishing problem itself, is a line of work on local pseudorandom generators, Goldreich-type functions, and algebraic attacks, often instantiated using sparse random or expanding graphs~\cite{DBLP:journals/rsa/MosselST06,DBLP:journals/cc/BogdanovQ12,DBLP:journals/joc/ApplebaumBR16,DBLP:journals/siamcomp/ApplebaumL18,DBLP:journals/siamcomp/Applebaum13}.

Positive results in this setting provide a polynomial-time algorithm when $M \gtrsim N^{k/2}$. For these values of parameters, the polynomial-time distinguisher between planted and null distributions is trivial (the birthday paradox implies that with decent probability there are two equations on exactly the same set of variables, and one can just check if those two have the same value), and a non-trivial algorithm can either recover a planted solution, or refute a random instance~\cite{COJA:OGHLAN_GOERDT_LANKA_2007, barak2016noisytensorcompletionsumofsquares,allen2015refuterandomcsp,STOC:RRS17,guruswami2023algorithmscertificatesbooleancsp}.

Moreover, some evidence for computational hardness is provided by the means of lower bounds in the restricted models of computation (Sum-of-Squares, Statistical Queries, low  degree polynomials) for both the distinguishing and refutation versions of the problem~\cite{CC:Grigoriev01, applebaum2010public, NIPS:FPV15,STOC:KMOW17,barak2023hiddenprogressdeeplearning,GHJS25}.

The connection between graph expansion properties of the constraint graph $H$ and the conjectured hardness of noisy $k$-XOR associated with $H$ is most transparent in terms of low-weight linear dependencies in the matrix $A_H$, or equivalently small even covers.
Graph expansion rules out the existence of small even covers, implying local pseudorandomness, lower bounds for low-degree distinguishers, and SoS lower bounds (see Section~\ref{sec:techover}).
These results, together with the lack of efficient algorithms for sparse instances, suggest the intuition that expansion may be the \emph{only} structural source of hardness for noisy $k$-XOR.
This intuition can be formalized into a mathematical conjecture in several ways, among which one of the strongest we provide below.

\begin{conjecture}

There is a constant $\alpha \in (1/2, 1)$, such that for every $k$-left regular bipartite graph $H= ([\constraints] \cup [\variables], E_H)$ that is $(T, \alpha)$-expanding, any circuit $C$ that succeeds in distinguishing the $\eta$-noisy $k$-XOR distinguishing problem with the constraint graph $H$, for $\eta\cdot T \gtrsim \log (\constraints)$
has size at least $2^{\Omega(T)}$.

\label{conj:weak}
\end{conjecture}
\begin{remark}
    See Lemma~\ref{lem:simple-distinguisher} for more discussion on why condition $\eta \gtrsim \log(M)/T$ is needed for the conjectured hardness.
\end{remark}

The belief that the corresponding noisy version of the $k$-XOR distinguishing problem is hard for every sufficiently expanding graph $H$ seems to go back at least to the work of Alekhnovich~\cite[Remark 1]{FOCS:Alekhnovich03} who proposes two YES and NO distributions in the spirit of our YES and NO distribution and says (in our notation) that ``We believe that (...), and if $H$ is an expander (which occurs
with probability $1 - O(1/n)$) then the distributions of
YES and NO are indistinguishable.''\footnote{See \Cref{rmrk:alek-distribution} for a description of actual YES and NO instances suggested by Alekhnovich, and a sketch of how our approach can be used to distinguish them.}

The exact conjecture that the $k$-XOR distinguishing problem as in~\Cref{def:k-xor-problem} is hard for every expanding graph $H$ was reiterated by Barak in~\cite[Page 39]{barak2014sum}, although without specifying a concrete setting of parameters.

A version of this conjecture with a more concrete realization of the parameters was suggested in~\cite{GHJS25}, albeit in the case of larger field $\bF_p$ for $p = \variables^{\Omega(1)}$, instead of $\bF_2$ as below (see Conjecture 4.3 therein). They used this conjecture, together with the planted-clique conjecture as a security basis for a Public-Key Encryption protocol.
Their conjecture is the following.

\begin{conjecture}
\label{conj:strong}
    There is a constant $\alpha \in (1/2, 1)$, such that for every $k$-left regular bipartite graph $H = ([\constraints] \cup [\variables], E_H)$ if it is $(T, \alpha)$-expanding, then every circuit that can solve the $\eta$-noisy $k$-XOR distinguishing problem must have size $M^{\omega(1)}$, where, with $\logVariables \coloneqq \log_2 \variables$,
    \begin{itemize}
        \item $k = \Omega(\logVariables)$,
        \item $T = \exp(\logVariables^\gamma)$ for some $\gamma \in (0,1)$,
        \item $\eta = \logVariables^{-\zeta}$ for some constant $\zeta$.
    \end{itemize}
\end{conjecture}

\subsection{Our Results}

In our work, we first refute Conjecture~\ref{conj:weak} by showing the following.

\begin{restatable}[See \Cref{thm:constant-rate}]{mainthm}{MainOne}\label{thm:main1}
   For every constant $\alpha > 0$ there is an infinite family of $k$-left regular constraint graphs $G = ([\constraints] \cup [\variables], E)$, where $\constraints = 2^{\Theta(\log^2 \variables)}$, and $k = (\log \variables)^{\Theta(1/\alpha)}$, which is $(\variables^{1-\alpha}, 1-o(1))$-expanding and there is an algorithm with running time $\poly(\constraints)$ to solve the $\eta$-noisy $k$-XOR distinguishing problem for those graphs, with constant noise rate $\eta = 1/3$.
\end{restatable}

Then, under a standard assumption that Reed-Muller codes efficiently achieve capacity for the binary erasure channel (BEC) (see Conjecture~\ref{conj:RM-random-erasure-strong} and a short discussion thereafter), we provide a construction with polynomial number of constraints.

\begin{restatable}[Informal, see \Cref{thm:polynomial-size-weak}]{mainthm}{MainTwo}\label{thm:main2}
Assume Reed-Muller codes efficiently achieve capacity for the BEC. 
Then, for every $\alpha,c>0$ there is an infinite family of
$k$-left-regular constraint graphs $H=([\constraints]\cup[\variables],E)$ such that
\[
\constraints = \variables^{\Theta(1)},
\qquad
k=(\log \variables)^{\Theta(1/\alpha)},
\]
the graph is $(\variables^{1-\alpha},1-o(1))$-expanding, and there is an algorithm
running in time $\poly(\constraints)$ that solves the $\eta$-noisy $k$-XOR distinguishing
problem on these graphs for $\eta = \variables^{-c}$.
\end{restatable}

In fact, we prove the conclusion of Theorem~\ref{thm:main2} under a significantly weaker conjecture on Reed-Muller codes: for some constants $\gamma, C > 0$, Reed-Muller codes of blocklength $M$ and rate $1-\varepsilon$ for $\varepsilon \gtrsim M^{-\gamma}$ can efficiently recover from $O(\varepsilon^C)$ random erasures (see Conjecture~\ref{conj:RM-random-erasure-weak}). 

\subsection{Our Techniques}
Our counterexample to expansion-based hardness in the context of the noisy XOR problem combines two previously separate lines of work.
Namely, we combine explicit lossless expander constructions from the pseudorandomness community with Reed-Muller decoding from random errors.

At a high level, the noisy XOR problem is naturally coding-theoretic, so using a decoding algorithm for a linear error-correcting code to disprove a conjecture about the hardness of recognizing a noisy XOR system is not surprising.
In particular, error correcting codes have been used to (weakly) refute other conjectures in the statistical-computational gap setting, most notably the low-degree conjecture in~\cite{ITCS:HW21,buhai2025quasipolynomiallowdegreeconjecturefalse}.
Yet, despite this obvious connection, to the best of our knowledge such conjectures on expansion-based hardness for noisy XOR have not previously been refuted in any nontrivial parameter regime.

The first obstacle is that the encoding matrix of an error-correcting code is hard to realize as the adjacency matrix of a sparse constraint graph.
Our transformation exploits some simple but specific properties of Reed-Muller codes to achieve this effect.
Namely, we observe that Reed-Muller decoding applies whenever the constraint graph is a \emph{coset graph}, that is, when it has additional algebraic structure such that every constraint defines an affine subspace (equivalently, a coset).

After establishing this connection, we draw on a powerful body of work from the 2010s showing that Reed-Muller codes can decode random errors far beyond their minimum distance~\cite{ASW14,SSV15,STOC:KKMPSU16}.
However, even then it is \emph{a priori} unclear how to ensure that the resulting graph is an expander.
Indeed, we leave open the question of whether a random subspace construction of coset graphs is expanding.

To overcome the second obstacle, namely ensuring that the encoding matrix comes from an expander, we turn to the explicit constructions of Guruswami, Umans, and Vadhan, who built lossless expanders from Parvaresh-Vardy codes~\cite{FOCS:PV05,GUV09}.
We show that their lossless expanders can be interpreted as coset graphs, allowing us to obtain expansion while retaining the ability to apply Reed-Muller decoding from random errors.

\subsection{Related Work on Hardness and Graph Expansion}  

Perhaps the earliest work in the deterministic\footnote{The nonlinear predicate for locally computing output bits produces noise that is a deterministic function of the input. In the noisy $k$-XOR setting we add \emph{random} Bernoulli error independently of the input.} noise setting that explicitly connects computational hardness and graph expansion properties is Goldreich's proposal of a candidate one-way function (OWF) in which a local predicate is evaluated on an expanding constraint graph and whose hardness heuristically depends on graph expansion~\cite{ECCC:Goldreich00,Goldreich11}.
The same construction idea for longer output lengths immediately gives a candidate pseudorandom generator (PRG).
Subsequently Oliveira, Santhanam, and Tell refuted this general conjecture by constructing a family of expander graphs such that instantiating Goldreich's PRG with this family is insecure even for reasonable predicate families~\cite{ITCS:OST19}.
In their work, they critically use the fact that the neighborhood function of the graph is of low complexity, e.g. affine or $\mathsf{AC}^{0}[\oplus]$. 
Our counterexample is incomparable with theirs for two reasons.
First, their counterexamples concern the deterministic noise setting while ours concern the random noise setting.
Secondly, while our construction also falls into a low-complexity regime because our neighbor function is affine, our algorithm for distinguishing noisy k-XOR proceeds through decoding from random erasures rather than complexity theoretic techniques.

A connection between graph expansion and hardness in the random noise setting was proposed in the work of Alekhnovich~\cite{FOCS:Alekhnovich03}.
Alekhnovich conjectured that the noisy $3$-XOR with exact error weight $w$ is indistinguishable from noisy $3$-XOR with exact error weight $w+1$ when the matrix $A$ is \emph{any} expander.
Our work morally refutes the conjecture (see Remark~\ref{rmrk:alek-distribution}).

Applebaum, Barak, Wigderson construct public-key encryption (PKE) from various combinations of three combinatorial assumptions whose hardness is related to expansion properties~\cite{DBLP:journals/siamcomp/Applebaum13}.
These constructions, however, use \emph{random} graphs that are expanders with high probability.
Similarly, a follow-up by Bogdanov, Kothari, Rosen combines ideas from the prior work to construct a PKE scheme whose security is based on Goldreich's PRG, thereby also using random graphs~\cite{TCC:BKR23}. 
Our counterexample has no immediate implications for either of these works.

More recently, the work of Ghosal, Hair, Jain, Sahai constructs PKE from the conjectured hardness of the planted clique problem and the noisy $k$-XOR over expander problem over large fields (see Conjecture~\ref{conj:strong})~\cite{GHJS25}.
The reason they use such a strong conjecture is that their PKE scheme constructs a \emph{structured} expander graph from a random $G(n, 1/2)$ graph.
While our counterexample does not address the case of large fields nor their specific structured expander family, it is cautionary counter evidence to statements such as their second assumption.

\subsection{Low degree method}
The \emph{low degree method} is a powerful heuristic that is gaining significant attention in the algorithmic learning community, widely employed to shed a light into computational complexity of concrete learning tasks.  A raising number of non-trivial algorithms in computational statistics is either accompanied, or followed by a ``matching'' lower bound against low-degree polynomials or algorithms in the statistical query model; and a lack of such lower bound is often a motivation to look for more efficient algorithms. 

While low-degree lower bounds themselves are often mathematically interesting, and provide deeper understanding of a structure of a problem at hand, the intuition that for ``natural'' and ``noisy enough'' problems the predictions provided by low-degree method should match the complexity of actual algorithms is sporadically challenged. One concrete example of this is the low-degree conjecture by Hopkins \cite{hopkins2018statistical}, a formal mathematical statement attempting to capture the intuition that every ``symmetric enough'' and ``noisy enough'' distribution, which is hard to distinguish from uniform by low-degree polynomials should be also hard to distinguish by all efficient algorithms. This conjecture has recently been weakly refuted in~\cite{buhai2025quasipolynomiallowdegreeconjecturefalse}.

Our construction does not satisfy the symmetry requirement of the low-degree conjecture, but it can be treated as providing a new example of a "noisy" problem for which low-degree predictions suggest exponential hardness, while a polynomial-time algorithm exists.

\section{Technical Overview}
\label{sec:techover}

We begin by reviewing how graph expansion implies lower bounds in several restricted computational models, e.g. SoS and low-degree polynomials.
This motivates our community's current belief that an expanding constraint graph $H$ should imply the computational hardness of the associated noisy $k$-XOR problem.
Then, we overview the construction of our counterexample, demonstrating that expansion cannot be the only explanation of hardness.

\subsection{Expansion Implies Lower Bounds}

We review explicitly the connection between graph expansion and known lower bounds.
These relations establish why the community has believed that graph expansion alone implies hardness for the noisy $k$-XOR problem.

\begin{definition}
In a bipartite graph $H = ([\constraints] \cup [\variables], E_H)$, we say that a set of left-vertices $C \subset [\constraints]$ is  \anevencover if every right vertex $x \in [\variables]$ has an even number of neighbors in $C$.
\end{definition}
This can be related to the standard notion of cycle in a graph $G$ in the following way. From a graph $G$ we build a $2$-left regular bipartite graph $H = ( E(G) \cup V(G), E_H )$
in which the left-vertices are the edges of $G$, and the right-vertices are the vertices of $G$ (in $H$ a left vertex $e$ is connected with a right vertex $v$ if $v$ lies on the edge $e$ in $G$). In this case, any minimal \evencover of $H$ corresponds exactly to a collection of edges forming a simple cycle in $G$.

Note that equivalently, in terms of linear algebra, \anevencover $C$ can be thought of as a subset of rows that add up to a zero vector, or $y \in \bF_2^\constraints$ such that $y^T A_H = 0$.

\begin{observation}
    If the constraint graph $H$ has \anevencover $z \in \bF_2^\constraints$ of size $\wt(z) \ll 1/\eta$, then the planted and null distributions are easy to distinguish.
\end{observation}
\begin{proof}
    Given vector $y$, consider $\langle z, y\rangle$. In the null distribution $\Pr[\langle z, y\rangle = 1] = \frac{1}{2}$,
    whereas in the planted distribution we have $\langle z, y \rangle = z^T (A_H x + e) = \langle z, e\rangle$.
    Now by union bound $\Pr[\langle z, e\rangle = 1] \leq \wt(z) \eta$.
\end{proof}

We can improve this bound by a logarithmic factor, if we have many disjoint \evencovers of small weight.
\begin{lemma}
    \label{lem:simple-distinguisher}
    If we have $S$ disjoint \evencovers $z_1, \ldots z_S$, each of weight at most $t$, we can distinguish planted from null distribution with noise level $\eta \lesssim \log(S)/t$.
\end{lemma}
\begin{proof}[Proof Sketch]
    If $z$ is \anevencover\hspace{-0.4em}, then $z^T y = z^T e$ is a sum (over $\bF_2$) of $t := \wt(z)$ independent random variables, each with $\Bern(\eta)$ distribution. We can interpret that $\Bern(\eta)$ random variable instead, as a variable that's $0$ with probability $(1-2\eta)$, and uniformly random with probability $2\eta$.

 The sum $\langle z, e\rangle$ is zero with probability $(1-2\eta)^t \approx \exp(-2\eta t)$, and uniformly random with the remaining probability. If we had $\exp(4 \eta t)$ independent samples (i.e. corresponding to disjoint \evencovers), we could distinguish the planted from the null distribution.
\end{proof}

Note that by simple linear algebra, every constraint graph $H = ([\constraints] \cup [\variables], E_H)$ has $\Omega(\constraints/\variables)$ pairwise disjoint \evencovers of length at most $\variables + 1$. As such, the noise level $\eta \leq \frac{\log(\constraints/ \variables)}{\variables}$ can be considered trivial.

The size of minimum even cover is relevant, when attempting to use the so-called \emph{low-degree heuristic} to understand the complexity of the noisy $k$-XOR problem: for any graph without small even cover the planted and null distributions are indistinguishable by low-degree polynomials.
In particular, if the smallest cover is of size at least $T + 1$, the planted distribution is $T$-wise independent and fools all degree $T$-polynomials over the reals.

Notably, a simple counting argument shows every sufficiently expanding constraint graph $H$ does not have small even covers.

\begin{definition}
    A set $S \subset [\constraints]$ of left vertices in a $k$-left-regular graph $H = ([\constraints] \cup [\variables], E_H)$ is said to be $\alpha$-expanding if
    \begin{equation*}
        |N(S)| \geq \alpha k |S|.
    \end{equation*}

    A $k$-left-regular graph $H$ is an $(T, \alpha)$-expander if for every set $S \subset [\constraints]$ of size at most $T$ is $\alpha$-expanding.
\end{definition}

\begin{lemma}
\label{lem:expander-has-no-small-circuit}
    If a graph $H$ is $(T, \alpha)$-expander for $\alpha > 1/2$, then every \evencover in $H$ is larger than~$T$.
\end{lemma}
\begin{proof}
Consider any set $S$ of size $|S| \leq T$. Since $H$ is an expander, $N(S) > k |S| / 2$. 

Consider an induced subgraph on $S \cup N(S)$. Clearly, adding up the degrees of left vertices, the total number of edges in this subgraph is just $k |S|$.

If every vertex in $N(S)$ had at least two neighbors in $S$, we can count the total number of edges in this graph by adding up the right-degrees of vertices in $N(S)$. We would have $|E[ S \cup N(S)]| \geq 2 |N(S)| > 2 k |S| /2 = k |S|$, a contradiction.
\end{proof}

The relevance of expansion is highlighted as a determining factor for a sum-of-squares lower bound for the noisy $k$-XOR problem. Concretely, Grigoriev showed that for a random $k$-regular constraint graph $H$, low-degree sum-of-squares cannot solve the $k$-XOR distinguishing problem \cite{CC:Grigoriev01}. Later, Barak \cite[Lemma 3.4]{barak2014sum} observed that in the Grigoriev's lower-bound proof, the only relevant property of a random $k$-regular constraint graph, was its expansion (see also \cite[Chapter 3.2]{baraksteurer16}). That is, he showed the following.

\begin{theorem}[\cite{barak2014sum}]
\label{thm:sos-lower-bound}
    For every constraint graph $H$ which is $(T, \alpha)$-expander, with $\alpha > 1/2$, Sum-of-Squares of degree much smaller than $T/100$ cannot distinguish between $A_H \cdot x$ (i.e. planted distribution with noise level $\eta = 0$) and the null distribution $\Unif(\bF_2^\constraints)$.
\end{theorem}

From this perspective, the Grigoriev's lower bound for a random $k$-XOR can be interpreted as a corollary of~\Cref{thm:sos-lower-bound} together with a standard fact that a random $k$-regular graph is highly expanding with high probability.

This expansion can also be used to prove lower bounds in various other restricted models of computation.

\begin{definition}[$\delta$-biased]
A distribution $D$ over $\{0,1\}^m$ is said to be $\delta$-biased if for every nonzero
$a \in \mathbb{F}_2^m$,
\[
\left|
\mathbb{E}_{y\sim D}\!\left[(-1)^{\langle a,y\rangle}\right]
\right|
\le \delta.
\]
\end{definition}

\begin{theorem}[Imported from~\cite{DBLP:journals/siamcomp/Applebaum13}, see Theorem 9.1 therein]
Let $A_H \in \mathbb{F}_2^{\constraints \times \variables}$ be a $k$-sparse\footnote{That is, every row has weight $k$.} matrix, and let
$H = ([\constraints] \cup [\variables], E)$ be its associated bipartite graph.
Assume that $H$ is a $(T,0.51)$-expander.
Let $y = A_H \cdot x + e$,
where $x$ is uniform in $\{0,1\}^{\variables}$ and
$e \in \{0,1\}^{\constraints}$ has independent $\mathsf{Ber}(\eta)$ coordinates.
Then the distribution of $y$ satisfies:
\begin{enumerate}
    \item it is $T$-wise independent;
    \item it is $\delta$-biased, where
    $\delta = \frac12(1-2\eta)^T$;
    \item for every Boolean function $f:\{0,1\}^{\constraints}\to\{0,1\}$ representable by a
    degree-$t$ polynomial over $\bF_2$,
    \[
    \left|
    \Pr[f(y)=1]-\Pr[f(u)=1]
    \right|
    \le
    8\cdot (1-2\eta)^{T/2^t - 1},
    \]
    where $u$ is uniform in $\{0,1\}^{\constraints}$.
\end{enumerate}
\end{theorem}

\noindent
Through a known result, this $T$-wise independence implies lower bounds against $\mathsf{AC}^0$ circuits~\cite{JACM:Braverman08}.  

Finally, since random $k$-left regular graphs are optimally expanding with large probability, all the lower bounds above can be used to deduce similar hardness for a random noisy $k$-XOR.

\begin{lemma}[{\cite[Theorem 11.2.3]{guruswami2012essential}}]
    There exists a universal constant $c$, such that the following holds.
    For every $\constraints, \variables$, where $\constraints \leq \variables^{ck}$, let us consider  a random $k$-left regular graph with $\constraints$ left vertices and $\variables$-right vertices. With  probability $1-o(1)$ this graph is a $(T, 3/4)$-expander, where $T \gtrsim \Omega(\constraints/k)$.
\end{lemma}

\subsection{Our Results}

First, combining our construction based on GUV expanders (\cite{GUV09}) with the reduction from decoding RM-codes over Binary symmetric channel (BSC) to decoding a (higher degree) code over binary erasure channel (BEC) \cite{SSV15} and the result that $\RM$ codes achieve capacity over BEC in the constant rate regime \cite{STOC:KKMPSU16}, we obtain the following theorem.

\MainOne*

This theorem statement refutes~\Cref{conj:weak}.
In the proof of this statement, we use a result of Kudekar, Kumar, Mondelli, Pfister, Şaşoğlu and Urbanke that RM codes achieve capacity over BEC, i.e. they recover from erasure rate
$\varepsilon - o_{\constraints}(1)$ where $\varepsilon \coloneq 1 - \mathsf{Rate}$~\cite{STOC:KKMPSU16}
\footnote{Following closely the proof in~\cite{STOC:KKMPSU16}, it is possible to extract an explicit upper bound on the \emph{gap to capacity} $o_M(1)$: RM codes of rate $1-\varepsilon$ can recover from $\varepsilon - O(1/\log \log \constraints)$ fraction of random erasures; see~\cite{raoyoutube} for an expository YouTube talk proving this bound.}.
To get a result for a polynomial number of equations in the number of variables we need much faster rate of convergence to capacity than is currently known.
We say that Reed-Muller codes efficiently achieve capacity for the binary erasure channel, if there is a constant $\gamma > 0$, such that any Reed-Muller codeword in $\RM(\logConstraints, r)$ over $\bF_2$ with rate $R = 1 -\varepsilon$ can be recovered with probability $1-o(1)$ from independent random erasures with erasure rate $\varepsilon - O(1/\constraints^\gamma)$. The constant $1/\gamma$ is known in the coding theory literature as a \emph{scaling exponent}.

\begin{conjecture}[Reed-Muller Codes Efficiently Achieve Capacity for BEC]
    \label{conj:RM-random-erasure-strong}
     There exists a constant $\gamma$, such that for every $\logConstraints, r \in \mathbb{N}$ such that $\RM(\logConstraints, r)$ has rate at most $1-\epsilon$,  a random erasure pattern sampled from $(\mathsf{Ber}(\eta))^{2^{\logConstraints}}$ is recoverable with probability $1 - o(1)$ provided that $\eta < \varepsilon - \frac{1}{\constraints^\gamma}$, where $\constraints = 2^\logConstraints$ is the blocklength of the underlying code.
\end{conjecture}

The question of the gap-to-capacity results for RM codes for symmetric channels has been explicitly raised in~\cite[Section V]{abbe2019reedmullercodespolarize}, \cite{hassani2018optimalscalingreedmullercodes, journal:mondelli2014polar}, and in~\cite{Journal:ASY21}. 

Note that this type of gap-to-capacity behavior is true and relatively simple to show for random codes, with a scaling exponent $1/\gamma = 2$, but it is much more challenging to prove for any explicit code with an efficient decoding algorithm. The breakthrough results~\cite{journal:guruswami-xia, journal:hassani2014finite} showed that polar codes efficiently achieve capacity for binary symmetric channels (i.e. have finite scaling exponent), and variants of polar codes can achieve $1/\gamma \to 2$ \cite{journal:gry22}.

\begin{remark}
Consider a linear space $\bF_2[X]_{\leq r}$ of low-degree polynomials with dimension $(1-\varepsilon) \constraints$, and a random subset $S$ of $(1 - \varepsilon) \constraints + \constraints^{1 - \gamma}$ locations on the hypercube $\bF_2^ \logConstraints$. \Cref{conj:RM-random-erasure-strong} is equivalent to a statement that with probability $1 - o(1)$ there is no non-zero polynomial in  $\bF_2[X]_{\leq r}$ that happens to vanish on all points of $S$.
\end{remark}

In fact, a plausibly simpler-to-prove conjecture suffices to derive a counterexample to the noisy-XOR distinguishing problem where the number of constraints $\constraints$ is polynomial in the number of variables $\variables$.

\begin{conjecture}[Weak Random Erasure Recovery]
    \label{conj:RM-random-erasure-weak}
     There exist constants $\gamma > 0$ and $\zeta\geq 1$, such that for every $\logConstraints, r \in \mathbb{N}$ if $\RM(\logConstraints, r)$ has rate at most $1-\epsilon$ for $\constraints^{-\gamma} <\varepsilon < 1/2$,  a random erasure pattern sampled from $(\mathsf{Ber}(\eta))^{2^\logConstraints}$ is recoverable with probability $1 - o(1)$ provided that $\eta < \varepsilon^{\zeta}$,  where $\constraints = 2^\logConstraints$ is the blocklength of the underlying code.
\end{conjecture}

\begin{remark}
    Conjecture~\ref{conj:RM-random-erasure-weak} is implied by Conjecture~\ref{conj:RM-random-erasure-strong}.
    Therefore, we state the following theorem only assuming Conjecture~\ref{conj:RM-random-erasure-weak}.
\end{remark}

Either of these conjectures holding would imply the following counterexample, in which the number of constraints $\constraints$ is polynomially related to the number of variables $\variables$, but the noise-rate is inverse polynomial.

\MainTwo*

Therefore, under Conjecture~\ref{conj:RM-random-erasure-strong} or Conjecture~\ref{conj:RM-random-erasure-weak}, we obtain an even stronger refutation of~\Cref{conj:weak}, in which we refute the setting where the number of constraints and variables are polynomially related, as opposed to quasi-polynomially related.

\begin{remark}
\label{rmrk:alek-distribution}
    The inverse-polynomial noise rate $\eta = \variables^{-c}$ was already suggested as a regime potentially hard for every expander $H$ in~\cite{FOCS:Alekhnovich03}.  The YES and NO distributions suggested there are slightly different: YES  instances consist of vectors $A_H x + e$ where $e$ is a random vector with sparsity exactly $t := \lfloor \eta \constraints \rfloor$, where NO distributions are given by vectors $A_H x + e$ where $e$ is a random $t+1$-sparse vector.
    
      Since the algorithm we propose for planted instances $y = A_H x + e$ actually recovers the solution $x$, our approach can be used to distinguish Alekhnovich's distributions as well. Specifically, upon receiving a vector $y$, we can independently flip each coordinate with probability somewhat larger than $\eta$, to obtain a vector $y' = A_H x + e + e'$. Now $e + e'$ is statistically close to a random vector with i.i.d. entries, hence with high probability we can correctly recover $A_H x$ from $y'$, and check if the Hamming distance between $A_H x$ and $y$ is $t$ or $t+1$.
\end{remark}

\section{Preliminaries}

We begin by recalling facts that relate the Hamming ball volume to the binary entropy function.
These facts will be used to prove correctness of our distinguisher.

\begin{fact}
\label{fct:hamming-volume}
    Let
    \begin{equation*}
        \binom{n}{\leq d} := \sum_{k \leq d} \binom{n}{k}.
    \end{equation*}
    Then for $d \leq n/2$, we have 
    \begin{equation*}
        n H_2(d/n) - o(n) \leq \log_2 \binom{n}{\leq d} \leq n H_2(d/n),
    \end{equation*}
    where 
    \begin{equation*}
        H_2(p) \coloneqq p \log (1/p) + (1-p) \log (1/(1-p))
    \end{equation*}
    is the binary entropy function.
\end{fact}

Now we recall standard asymptotic bounds on the binary entropy function.

\begin{fact}[Standard Entropy Bounds]
Let $H_2$ be the binary entropy function.
Then
\begin{enumerate}
    \item For $\varepsilon \to 0$, 
    \[H_2 \left(\varepsilon \right) = O(\varepsilon\log (1/\varepsilon)).\]
    \item For $\varepsilon \to 0$, \[H_2 \left( \frac12 - \varepsilon \right) = 1 - \Theta(\varepsilon^2).\]
\end{enumerate}
\end{fact}

Finally, we recall the Berry-Ess\'een Theorem.
It allows us to obtain more precise lower bounds on $1-R$ for the rate $R$ of a $(\logConstraints, r)$ Reed-Muller code where $r > \logConstraints/2$.

\begin{theorem}[Berry-Ess\'een Theorem~\cite{tao2023topics}]
\label{thm:berry-esseen}
Let $X$ have mean zero, variance one and finite third moment.
Let $X_1, \ldots, X_n$ be  i.i.d. copies of $X$, and let $S_n \coloneqq \left( X_1 + \cdots + X_n \right) / \sqrt{n}$.
Then,
\[\Pr \left[ S_n  < a\right ] \leq \Pr \left[ Z  < a\right ] + \frac{  K \cdot \E \left[ |X|^3 \right ] }{ \sqrt{n} } \]
    uniformly for all $a \in \bR$ where $Z \sim N(0, 1)$ and $K$ is some absolute constant. 
\end{theorem}

\section{Our Counterexample}

 We will discuss a generic construction of a family of graphs for which $A_H \cdot  x$ is a Reed-Muller codeword; as such, in some noise range the distinguishing problem can be solved in polynomial time by the known results for decoding Reed-Muller codewords from random errors. 

What remains to be shown is that there is a graph in this class which is a $(T, 3/4)$-expander for some decent value of $T$. 

\subsection{Graphs with Reed-Muller decoding}

In order to use the theory of Reed-Muller codes for the distinguishing problem, we will use an additional structure of a constraint graph.

\begin{definition}[Coset graph] Given $k$ matrices $A_1, \ldots A_k \in \bF_2^{d \times \logConstraints}$ (where $d < \logConstraints$), the $(\logConstraints,k,d)$-coset graph associated with $(A_1, \ldots A_k)$ is the following:
\begin{itemize}
    \item The set of left vertices $[\constraints]$ is identified with points on a hypercube $\bF_2^\logConstraints$.
    \item The set of right vertices $[\variables]$ is identified with $[k] \times \bF_2^{d}$
    \item Left vertex $v \in \bF_2^\logConstraints$ is connected with the right vertex $(i, u) \in [k] \times \bF_2^{d}$ if $A_i \cdot v = u$.
\end{itemize}
Alternatively, for a collection $(V_1, \ldots V_k)$ of subspaces of dimension at least $\logConstraints-d$ in $\bF_2^\logConstraints$, the $(\logConstraints,k,d)$ coset graph associated with $(V_1, \ldots V_k)$ is the $(\logConstraints,k,d)$-coset graph associated with $(A_1, \ldots A_k)$ where matrices $A_i$ are such that $\ker A_i = V_i$ (the coset graph does not depend on the choice of matrices $A_i$, only on the subspaces $V_i$).
\end{definition}

We now observe that the code generated by the columns of the matrix $A_H$ forms a subcode of the Reed-Muller code $\RM(\logConstraints, d)$.

\begin{lemma}[Coset Graphs give RM Subcodes]
    \label{lem:coset-graph-is-subcode}
    Let $H$ be an $(\logConstraints, k, d)$-coset graph and let $A_H \in \mathbb{F}_2^{2^\logConstraints \times k \cdot 2^d}$ be the adjacency matrix.
    Then the column space $\mathsf{Im}(A_H) \subseteq \mathbb{F}_2^{2^{\logConstraints}}$ is a subcode of $\RM(\logConstraints, d)$.
\end{lemma}
\begin{proof}
We will show that every column of $A_H$ is a $\RM(\logConstraints, d)$ codeword, i.e. an evaluation vector of a multilinear polynomial of degree at most $d$ on all of $\bF_2^\logConstraints$.
Every column of $A_H$ corresponds to a right vertex of $H$ given by a pair $(i, u)$ with $i \in [k]$ and $u \in \bF_{2}^{d}$.
Its neighborhood of left vertices is given by the affine subspace
\[W_{i, u} \coloneqq \{ v \in \bF_2^{\logConstraints} : A_i \cdot v = u \}.\]
where by the definition of the $(\logConstraints, k, d)$-coset graph, we have $\rank(A_i) \leq d$.
In the adjacency matrix $A_H$, for every $v \in \bF_2^{\logConstraints}$, the $v$-th entry of the $(i, u)$-column of $A_H$ is $1$ if and only if $A_i \cdot  v = u$.
Observe that this subspace can be characterized by at most $d$ many affine constraints of the following form for linear functions $\ell_j : \bF_{2}^\constraints \to \bF_{2}$, given by the $j$-th row of $A_i$, and scalar values $u_j$ over $\bF_2$:
\[\left \{ \ell_j(X) + u_j = 0  \right \}_{j \in [\rank(A_i)]}.\]
Therefore, the evaluation vector is exactly that of a degree $\rank(A_i) \leq d$ indicator polynomial for the subspace $W_{i, u}$:
\[\mathbf{1}_{i, u}(X) = \prod_{j \in [\rank(A_i)]} \left( \ell_j(X) + u_j + 1\right).\]
Since every column of $A_H$ is a codeword of $\RM(\logConstraints, d)$, the column span $\Image(A_H)$ is contained in $\RM(\logConstraints, d)$.
\end{proof}

Therefore, if we can decode the Reed-Muller code up to the $\eta$ fraction of random errors, we can distinguish the $A_H \cdot x + e$ vector from a random one for any coset graph $H$.

Let us first see what we can obtain using the unique decoding of the $\RM$ codes.
\begin{fact}
    The distance of the Reed-Muller code $\RM(\logConstraints, r)$ is $\Delta = 2^{\logConstraints-r}$. Relative distance is $\delta = 2^{-r}$.
\end{fact}

Since we can uniquely decode $\RM$ codes up to an error rate $\delta/2$, we have the following.
\begin{corollary}
    For every $(\logConstraints,k,d)$-coset graph, the associated XOR-distinguishing problem can be solved in polynomial time when $\eta < \delta / 2 \approx \frac{k}{\variables}$. 
\end{corollary}
Note that the largest $\alpha$-expanding set is in a $k$-left regular bipartite graph on $[\constraints] \cup [\variables]$ is of size proportional to $\variables/k$. Even in the most optimistic scenario with respect to the expansion of random $(\logConstraints,k,d)$-coset graphs, this falls (barely) short of disproving the \Cref{conj:weak}, as we would prefer the noise rate to be at least larger by a $\log(\variables)$ factor. This is not particularly helpful; we can try to improve on that in the following way.

\subsection{Known Results in RM Decoding from Random Erasures}
The approach for distinguishing a randomly corrupted codeword from a uniformly random string will leverage the fact that corruptions are introduced at random. 
Several prior works~\cite{ASW14,SSV15,SS18} have studied decoding Reed-Muller codes with random errors (see also \cite{Journal:ASY21} for an exposition of recent results on Reed-Muller codes). 
Concretely, \cite{SSV15} provides an algorithmic result for decoding Reed-Muller codes from random errors of a rate vastly exceeding the distance of the code. Specifically, they show the following.

\begin{theorem}[\cite{SSV15}]
    \label{thm:efficient-decoding}
    There is an efficient\footnote{By efficient, we mean polynomial time in the block length $2^\logConstraints$.} algorithm that corrects a pattern $P \subset [\constraints]$ of corruptions in $\RM(\logConstraints,r)$ codeword, if the corresponding codeword can be recovered from the same pattern of \emph{erasures} $P$ in a Reed-Muller code of higher degree $\RM(\logConstraints, \frac{\logConstraints+r}{2})$.
\end{theorem} 

This reduces the question of efficient correction of random corruptions, to a purely analytical question of understanding what fraction of erasures the Reed-Muller code can reconstruct under the binary erasure channel. 
To this end, the work of Kudekar et al. implies the following theorem statement.

\begin{theorem}[Random Erasure Recovery~\cite{kkmpsu16}]
    \label{thm:RM-random-erasure}
    For every $\logConstraints, r \in \mathbb{N}$ such that $\RM(\logConstraints, r)$ has rate at most $1-\epsilon$,  a random erasure pattern sampled from $(\mathsf{Ber}(\eta))^{2^\logConstraints}$ is recoverable\footnote{By recoverable we mean that the codeword is information-theoretically determined. For any linear code this also implies that it can be efficiently recovered by solving a linear system.} with probability $1 - o(1)$ 
    provided that $\eta < \varepsilon - o(1)$.
\end{theorem}

Armed with these two theorems, we are ready to prove that this decoder leads to a distinguisher.
Before doing so in Section~\ref{sec:decoding}, we first identify an exact expanding family of coset graphs.

\subsection{GUV Expander from PV Codes}

We will use a brilliant construction of~\cite{GUV09} of an unbalanced expander graph from Parvaresh-Vardy codes.
First, we will show that this expander graph family is a $(\logConstraints, k, d)$ coset graph, giving a Reed-Muller subcode.
Then, we will give two parameter regimes of interest---the first case being when $d = O(\sqrt{\logConstraints})$ and the second being when $d = O(\logConstraints)$.
These two settings correspond respectively to the settings of a quasi-polynomial number of constraints and polynomial number of constraints.
\\

For fixed $E \in \mathbb{F}_Q[X]_{=\pvDeg+1}$ an irreducible polynomial of degree $\pvDeg+1$, and a prime power $Q$, we define a $(E, Q, \pvDeg, \pvLen, h)$-GUV graph as follows:
\begin{itemize}
\item The set of left vertices is identified with the set of all polynomials over $\mathbb{F}_Q$ of degree at most $\pvDeg$: $\mathbb{F}_Q[X]_{\leq \pvDeg} \approx \mathbb{F}_Q^{\polyLength}$. 
\item The set of right vertices is identified with $\mathbb{F}_Q \times \mathbb{F}_Q^\pvLen$.
\item Each left-vertex has exactly $Q$ neighbors: one in each $\{y\} \times \mathbb{F}_Q^\pvLen$ for every $y \in \mathbb{F}_Q$.
\item Specifically, for a polynomial $f$, and $y \in \mathbb{F}_Q$, the $y$ neighbor of $f$ is $(y, f_0(y), \ldots f_{\pvLen-1}(y))$ where $f_i := f^{h^i} \bmod E$.
\end{itemize}
 
Note that when $h$ and $Q = 2^{q}$ are powers of two, we can identify $\mathbb{F}_Q^{\polyLength}$ with $\mathbb{F}_2^{q \polyLength }$ while preserving the additive structure. In this sense $f \mapsto f^{2^i}$ is $\mathbb{F}_2$ linear and $f \mapsto f \pmod E$ is $\mathbb{F}_Q$ linear (hence, also $\mathbb{F}_2$ linear), as is the evaluation map $f \mapsto f(y)$, so the composition from $\mathbb{F}_2^{q \polyLength} \to \mathbb{F}_2^{(\pvLen+1)q}$ sending $f$ to $(y,f_0(y), \ldots f_{\pvLen-1}(y))$ is $\mathbb{F}_2$ linear: a preimage of a point is an affine subspace of codimension at most $(\pvLen+1)q$. 
\begin{fact}
\label{fct:guv-is-coset}
    Let $Q = 2^q$ and $h = 2^t$ be powers of two for $q, t \geq 1$.
    Then every $(E, Q, \pvDeg, \pvLen, h)$-GUV graph is an $(\logConstraints,k,d)$-coset graph with $\logConstraints = (\pvDeg + 1)q$, $k = Q$ and $d = \pvLen \cdot q$.
\end{fact}
\begin{proof}
    Fix an $\mathbb{F}_2$-basis of $\mathbb{F}_Q$ where we view $\mathbb{F}_Q \cong \mathbb{F}_2^{q}$ as a $\mathbb{F}_2$-vectorspace. 
    Then the set $\mathbb{F}_Q[X]_{\leq \pvDeg}$ is a $\mathbb{F}_2$-vectorspace of dimension $\logConstraints \coloneqq (\pvDeg+1)q$.
    Consider the map
    \begin{align*}
    &L_y : \mathbb{F}_Q [X]_{\leq \pvDeg} \to \mathbb{F}_Q^{\pvLen}, &f \mapsto (f_0(y), \ldots, f_{\pvLen-1}(y))
    \end{align*}
    where $f_i \coloneqq f^{h^{i}} \bmod E$.
    Then the GUV graph has edges between left vertex $f \in \mathbb{F}_Q[X]_{\leq \pvDeg}$ and right vertex $(y \in \mathbb{F}_Q, z \in \mathbb{F}_Q^{\pvLen})$ if and only if $L_y(f) = z$.

    Observe that $f \mapsto f^{h^{i}}$ is $\mathbb{F}_2$-linear because for any integer $\ell \geq 0$ the following holds in characteristic two field arithmetic:
    \begin{align*}
        (f + g)^{2^{\ell}} &= f^{2^{\ell}} + \left ( \sum_{i=1}^{2^{\ell} - 1} \binom{2^{\ell}}{i} f^{2^{\ell} - i} g^{i} \right )+ g^{2^{\ell}} \\
        &= f^{2^{\ell}} +  g^{2^{\ell}}.
    \end{align*}
    Then recall that modular reduction by irreducible $E$ and evaluation at a fixed point $y$  are both $\mathbb{F}_Q$-linear, therefore $\mathbb{F}_2$-linear.
    Therefore, we can represent $L_y$ by a binary matrix $A_y \in \mathbb{F}_2^{\pvLen \cdot q \times \logConstraints}$.
    Then, for any left vertex $f$ written in $\mathbb{F}_2^{(\pvDeg+1)q}$ representation and right vertex $z$ written in $\mathbb{F}_2^{\pvLen \cdot q}$ representation, we have the edge condition above given exactly by the condition $A_y \cdot f = z$.
    These matrices $\{A_y\}_{y \in \mathbb{F}_2^{q}}$ define a coset graph with $\logConstraints = (\pvDeg+1)q, k = Q$, $d = \pvLen \cdot q$.
\end{proof}

The expansion of GUV graphs has been famously analyzed by~\cite{GUV09} (see also \cite[Theorem 5.35]{vadhan2012}).
\begin{theorem}[\cite{GUV09}]
\label{thm:guv-expansion}
    Every $(E, Q, \pvDeg, \pvLen, h)$-GUV graph is an $(T, \alpha)$-expander for $T = h^\pvLen$ and $\alpha = 1 - \frac{\pvDeg \cdot \pvLen \cdot h}{Q}$.
\end{theorem}

\begin{corollary}
    \label{cor:expanding-coset-family}
    For every $\alpha \in (0,1)$ and $\gamma \in (0, 1)$, there is an infinite family of explicit $(\logConstraints,k,d)$-coset graphs, with $d = \gamma \logConstraints + \Theta(\log \logConstraints)$ and $k = (\logConstraints/\lg \logConstraints)^{\Theta(1/\alpha)}$, each of which is $(T, 1-o(1))$-expanding with $T = 2^{(1-\alpha)d}$.
\end{corollary}
\begin{proof}
Fix constants $\alpha, \gamma \in (0, 1)$. Define parameter $D \in \mathbb{N}$ that serves as an index for the infinite family of coset graphs.
Let $C > 2$ be a constant and define the following parameters:
\begin{align*}
    &q \coloneqq \left \lceil \frac{C}{\alpha} \lg D \right \rceil, &Q \coloneqq 2^q,\\
    &\logConstraints \coloneqq qD, &\constraints \coloneqq 2^{\logConstraints},\\
    &\logVariables \coloneqq \left \lfloor \gamma D \right \rfloor, &\variables \coloneqq 2^{(\logVariables+1) q},\\
    & t \coloneqq \left \lceil  \frac{(\logVariables+1)(1-\alpha)}{\logVariables} \cdot q\right \rceil, & h \coloneqq 2^{t}.
\end{align*}
For field size $Q$ and for any degree $D$ irreducible polynomial $E(X) \in \mathbb{F}_Q[X]$, we have a $(E, Q, D-1, \logVariables+1, h)$-GUV graph with $\constraints$ left nodes and $\variables$ right nodes with $Q$ left regular degree. 
By Fact~\ref{fct:guv-is-coset}, this GUV graph is a $(\logConstraints, Q, \log \variables)$-coset graph where
\[\log \variables = (\logVariables+1) \cdot q = \gamma \cdot \log \constraints  + \Theta(\log \log \constraints).\]
Moreover, we have $\logConstraints = \Theta\left (\frac{D \lg D}{\alpha} \right)$, implying that $D = \Theta(\alpha \logConstraints / \lg \logConstraints)$. 
Therefore,
\[Q = D^{\Theta(1/\alpha)} = (\logConstraints / \lg \logConstraints)^{\Theta(1/\alpha)}.\]
Now consider the expansion factor.
First observe that our choice of $q$ directly implies that $Q^{\alpha} \geq D^{C}$ for $C > 2$.
Then, observe $h = Q^{1 - \alpha + o(1)}$ which together with the above implies that
\[\frac{(D-1) \cdot (\logVariables+1) \cdot h}{Q} = O(D^2 \cdot Q^{- \alpha + o(1)}) = o(1).\]
Finally, observe that our choice of $t$ gives $h^{\logVariables} \geq \variables^{1-\alpha}$.
By Theorem~\ref{thm:guv-expansion}, this coset graph is $(\variables^{1-\alpha}, 1 - o(1))$ expanding.
\end{proof}

\begin{remark}
\label{rmk:rational-gamma}
    For rational $\gamma$ and appropriate choices of $D$ such that $D \gamma$ is integral, we can sharpen the theorem statement to $d = \gamma \logConstraints$, removing the additive $\Theta(\log \logConstraints)$ term.
\end{remark}

\begin{remark}
\label{rmk:sqrt-d}
    The same theorem statement holds for an infinite family of explicit $(\logConstraints, k, d)$-coset graphs with $d = c_\beta \cdot \sqrt{\logConstraints}$ and $k = (\logConstraints / \lg \logConstraints)^{\Theta(1/\alpha)}$ where $c_\beta$ can be an arbitrarily small constant that depends on a constant parameter $\beta$.
    To see this, observe that the proof above holds when the parameter $\logVariables = \left \lfloor \beta \cdot \sqrt{D / \log D} \right \rfloor$ where $D$ was the index for the infinite family and $\beta$ is any constant of our choice. 
\end{remark}

\section{\texorpdfstring{Distinguisher for the Noisy $k$-XOR Problem on Coset Graphs}{Distinguisher for the Noisy k-XOR Problem on Coset Graphs}}
\label{sec:decoding}

We directly combine the two known theorems, Theorem~\ref{thm:efficient-decoding} and Theorem~\ref{thm:RM-random-erasure}, to construct a distinguisher running in polynomial time in the blocklength $\constraints = 2^\logConstraints$ for the setting in which the error rate $\eta$ has a $o(1)$ gap to the capacity.
The distinguisher proceeds in two steps.
First, it uses Reed-Muller decoding to recover a candidate codeword.
Then, it checks if the corresponding error is of the expected weight.
The following result is unconditional.

\begin{theorem}
\label{thm:decoding-coset-graph}
For any constant $c \in (0, 1)$, for any $ d = d(\logConstraints) < \logConstraints$ such that $\logConstraints + d$ is an even integer and such that $p = d/\logConstraints < c$ for all sufficiently large $\logConstraints$, let $\varepsilon_{\logConstraints, d} \coloneqq 1 - \Rate\left( \RM \left(\logConstraints, \frac{\logConstraints+d}{2} \right ) \right)$.
For any $\eta = \eta(\logConstraints)$ satisfying  $\eta <  \varepsilon_{\logConstraints, d} - o(1)$  and
\[
        \eta \cdot 2^{\logConstraints}  \lesssim 2^{\logConstraints \cdot H_2((1 - p)/2)},
\]
for every $k  = (\log \variables)^{O(1)}$, for any $(\logConstraints,k,d)$-coset graph $H$, there exists an algorithm running in time polynomial in the blocklength $\constraints = 2^\logConstraints$, that solves the $\eta$-noisy XOR distinguishing problem for $H$ with distinguishing advantage $1 - o(1)$. 
\end{theorem}
\begin{proof}
    For convenience, let $\constraints \coloneqq 2^\logConstraints$.
    By Theorem~\ref{thm:efficient-decoding}, any error pattern that is recoverable from erasures in $\RM \left (\logConstraints, \frac{\logConstraints+d}{2} \right)$ is also decodable from errors in $\RM(\logConstraints,d)$. 
By Theorem~\ref{thm:RM-random-erasure}, if
\[
\eta < \varepsilon_{\logConstraints,d} - o(1)
\]
then a random erasure pattern sampled from $(\Ber(\eta))^{\constraints}$ is (possibly inefficiently) decodable in $\RM \left (\logConstraints,\frac{\logConstraints+d}{2} \right ) $ with probability $1-o(1)$. 
Therefore by
Theorem~\ref{thm:efficient-decoding}, there is an algorithm
$\mathcal{A}_{\RM}$ that runs in time polynomial in the block length $2^\logConstraints$ such that for every $c\in \RM(\logConstraints,d)$,
\[
\Pr_{e\sim(\Ber(\eta))^\constraints} \left[\mathcal A_{\RM}(c+e)=c\right]=1-o(1)
\]
Let $H$ be an $(\logConstraints,k,d)$-coset graph, and define the subspace
\[
C_H \coloneqq \Image(A_H)\subseteq \bF_2^\constraints.
\]
By Lemma~\ref{lem:coset-graph-is-subcode}, we have $C_H\subseteq \RM(\logConstraints,d)$ is a subcode.
Therefore, for every $c\in C_H$,
\[
\Pr_{e \sim (\Ber(\eta))^\constraints} \left[ \mathcal{A}_{\RM}(c+e)=c\right ]=1-o(1).
\]
Define the distinguisher $\Dist$ as follows. On input $y \in \bF_2^\constraints$:
\begin{enumerate}
    \item Compute $\hat{c} \xleftarrow{} \mathcal A_{\RM}(y)$.
    \item Output $1$ if and only if $\wt(y - \hat{c}) \leq t$ and $\hat{c} \in C_H$.
\end{enumerate}
where $t= (1 + \delta) \eta \constraints$ for a small constant $\delta > 0$ is such that standard tail bounds imply that
\[
\Pr_{e \sim (\Ber(\eta))^\constraints} [ \wt(e) \leq t] = 1 - o(1).
\]
This distinguisher is polynomial-time because $A_{\RM}$ is a polynomial-time algorithm and the codeword check can be done via linear algebra. 

In the planted case, $y = c + e$ for a uniformly random $c\in C_H$ and
$e \sim (\Ber(\eta))^\constraints$. 
With probability $1 - o(1)$ over the choice of $e$, the decoder returns
$\hat{c}=c$ and $\wt(e)\leq t$. 
Hence
\[
\Pr[\Dist(y)=1]=1-o(1).
\]
Now consider the null case where the input is $y\sim\Unif(\bF_2^\constraints)$. 
We'll show the probability that $\Dist$ accepts such an input is $o(1)$.
If $\Dist(y)=1$, then by definition of $\Dist$,
there exists some codeword $\hat c\in \RM(\logConstraints,d)$ with $\wt(y-\hat c)\leq t$ such that $\hat{c} \in C_H$.
Let $\Vol(\constraints, t) \coloneqq \binom{\constraints}{\leq t} $ denote the Hamming ball of radius $t$ in dimension $\constraints$. 
By a union bound over all codewords of $C_H$, the probability that $y$ is within $t$ Hamming weight of any codeword in $C_H$ is
\[ \Pr \left [d(y,C_H)\leq t \right ] \leq \frac{|C_H| \cdot \Vol(\constraints,t)}{2^\constraints}.\]
First observe that since $k = (\log N)^{O(1)}$,
\[\lg(|C_H|)=\rank(A_H) \leq k 2^d = \constraints^{p + o(1)} \leq \constraints^c = o(\constraints).\]
Then recall that $t$ can be taken to be $(1 + \delta)\cdot \eta \cdot \constraints $ for any arbitrarily small constant $\delta$.
Since $p \in (0, c)$ for $c < 1$ and $\eta \cdot \constraints \lesssim 2^{\logConstraints \cdot H_2((1-p)/2)}$, we have that $t = o(\constraints)$. 
Then, $H_2(t/\constraints) = o(1)$ and Fact~\ref{fct:hamming-volume} implies
\[\lg (\Vol(\constraints, t)) \leq {\constraints \cdot H_2(t/\constraints)} = o(\constraints).\]
Therefore, 
\[\Pr \left [d(y,C_H)\leq t \right ] = 2^{-\constraints + o(\constraints)} = o(1).\]
Thus, we conclude that
\[
\Pr_{y \sim \Unif(\bF_2^\constraints)}[\Dist(y)=1]=o(1),
\]
so $\Pr[\Dist(y)=0]=1-o(1)$ in the null case.
\end{proof}

As previously, we can inspect the consequences of this in range $d = \gamma \logConstraints$ and $d = (1-\gamma) \logConstraints$ for small constant $\gamma$~---this theorem leads to the same range of parameters as the distinguishing one discussed in the previous section, but with a dual guarantee.

\begin{corollary}
    When $d = (1-\gamma) \logConstraints$, the bound in~\Cref{thm:decoding-coset-graph}  simplifies to
    \begin{equation*}
        2^\logConstraints \eta \lesssim 2^{\logConstraints H_2(\gamma/2)} = (2^{m-d})^{\Omega(\log 1/\gamma)} = \Delta^{\Omega(\log 1/\gamma)}
    \end{equation*}
    for $\gamma \to 0$ and where $\Delta = 2^{\logConstraints-d}$ is the distance of the underlying code.
\end{corollary}
\begin{proof}
     Recall the fact that for $\gamma \to 0$, 
     \[H_2(\gamma) = \gamma \log (1/\gamma) + (1-\gamma) \log (1/(1-\gamma)) = \Theta(\gamma \log (1/\gamma)).\] Then observe that $\Delta = 2^{\gamma\cdot  \logConstraints}$ and the rest follows directly by substitution.
\end{proof}

\begin{corollary}
    When $d = \gamma \logConstraints$, the bound in~\Cref{thm:decoding-coset-graph}  simplifies to
    \begin{equation*}
        \eta \lesssim 2^{- \Omega(\gamma^2 \logConstraints)} = \delta^{\Omega(\gamma)}
    \end{equation*}
    for $\gamma \to 0$, where $\delta = 2^{-\gamma \logConstraints}$ is the relative distance of the underlying code.
\end{corollary}
\begin{proof}
Observe that $(1-p)/2 = (1-(d/\logConstraints))/2 = (1/2) - (\gamma/2)$.
Taking the Taylor expansion of the binary entropy function  around $1/2$ gives
$H_2((1/2) - (\gamma/2)) = 1 - \Theta(\gamma^2)$.
Substitution then gives the desired result. 
\end{proof}

The best possible expanding set is of size $\variables/k = 1/\delta$. If we had $(T, 3/4)$-expanding coset graph with $T = (\variables/k)^{1-\alpha}$ (as we could hope from the PV codes), this leads to $T \eta \approx \delta^{-1 + \alpha + \Omega(\gamma)}$. As long as $\alpha + C\gamma \leq 1/2$, this noise rate is enough to solve the conjectured one.

\begin{theorem}
\label{thm:decoding-coset-graph-weak}
Assume Conjecture~\ref{conj:RM-random-erasure-weak}. Let $\xi >0$ and $\zeta\ge 1$
be the constants from that conjecture.
For any constant $c\in(0,1)$, for any $d=d(\logConstraints)<\logConstraints$ such that $\logConstraints+d$ is an even integer and such that $p \coloneqq d/\logConstraints < c$
for all sufficiently large $\logConstraints$, let
\[
\varepsilon_{\logConstraints,d}\coloneqq 1-\Rate \left(\RM \left(\logConstraints,\frac{\logConstraints+d}{2}\right)\right),
\qquad \qquad
\constraints\coloneqq 2^\logConstraints.
\]
If $\constraints^{-\xi}<\varepsilon_{\logConstraints,d}<\frac12$,
then, for any $\eta=\eta(\logConstraints)$ satisfying
\[
\eta<\varepsilon_{\logConstraints,d}^{\zeta}
\qquad\text{and}\qquad
\eta \constraints \lesssim \constraints^{H_2((1-p)/2)},
\]
for every $k = (\log N)^{O(1)}$, for any $(\logConstraints,k,d)$-coset graph $G$, there exists a $\poly(\constraints)$ time algorithm that solves the $\eta$-noisy XOR distinguishing problem for $G$ with distinguishing advantage $1 - o(1)$.
\end{theorem}

\begin{proof}
    The proof is identical to that of Theorem~\ref{thm:decoding-coset-graph},
with Conjecture~\ref{conj:RM-random-erasure-weak} replacing
Theorem~\ref{thm:RM-random-erasure} in the proof of Theorem~\ref{thm:decoding-coset-graph}.
\end{proof}

\section{Putting Decoding and Expansion Together}

We now restate and prove our main theorems.
Our first main theorem is for the case where the number of constraints is quasi-polynomial in the number of variables and holds unconditionally with constant error rate.

\begin{theorem}[Main Theorem~\ref{thm:main1}]
\label{thm:constant-rate}
    For every constant $\alpha > 0$ there is an infinite family of $k$-left regular constraint graphs $G = ([\constraints] \cup [\variables], E)$, where $\constraints = 2^{\Theta(\log^2 \variables)}$, and $k = (\log \variables)^{\Theta(1/\alpha)}$, which is $(\variables^{1-\alpha}, 1-o(1))$-expanding and there is an algorithm with running time $\poly(\constraints)$ to solve the $\eta$-noisy $k$-XOR distinguishing problem for those graphs with distinguishing advantage $1 - o(1)$, with constant noise rate $\eta = 1/3$.
\end{theorem}

\begin{proof}

Applying Remark~\ref{rmk:sqrt-d} to Corollary~\ref{cor:expanding-coset-family}, there is an infinite family of explicit $(\logConstraints,k,d)$-coset graphs
\[
G_\logConstraints=([\constraints]\cup [\variables],E_\logConstraints)
\]
that is $(\variables^{1-\alpha},1-o(1))$-expanding with the following relation on parameters:
\begin{itemize}
    \item $\constraints=2^\logConstraints$,
    \item $d= c \cdot \sqrt{\logConstraints}$ for $c = 0.1$ (chosen specifically for $\eta = 1/3$).
    Observe that $c$ satisfies two conditions:
        \begin{equation}
             \Pr \left[ Z \leq c/2 \right ] < 0.6, \label{eqn:hypo-1}
        \end{equation}   
        for a standard normal $Z \sim N(0, 1)$ and 
        \begin{equation}
            c^2 < -\ln (1/3)/ 2 \label{eqn:hypo-2}
        \end{equation}
    \item $k=(\logConstraints/\log \logConstraints)^{\Theta(1/\alpha)}$.
\end{itemize}
First, we claim that $\constraints = 2^{\Theta(\log^2 \variables)}$, and $k = (\log \variables)^{\Theta(1/\alpha)}$.
Since $G_\logConstraints$ is an $(\logConstraints,k,d)$-coset graph, its right side has size $\variables=k\cdot 2^d.$
Because $d=\Theta(\sqrt{\logConstraints})$ and $\log k= O(\log \logConstraints)$, we have $\log \variables=d+\log k= \Theta(\sqrt{\logConstraints})$, 
implying $\logConstraints=\Theta \left(\log^2 \variables \right)$.
Therefore,
\[ \constraints=2^\logConstraints=2^{\Theta(\log^2 \variables)}. \]
Also, since $\logConstraints=\log \constraints$, the left degree satisfies\footnote{Since $\logConstraints = \Theta(\log^2 \variables)$, for any constant $\varepsilon > 0$ and sufficiently large $\variables$, $(\log \variables)^{(2-\varepsilon)/\alpha} \leq \logConstraints/\log \logConstraints  \leq (\log \variables)^{2/\alpha}$.}
\[ 
k=(\logConstraints/\log \logConstraints)^{\Theta(1/\alpha)}=(\log \variables)^{\Theta(1/\alpha)}.
\]

Now we check that our choice of parameters satisfy the hypotheses in the statement of Theorem~\ref{thm:decoding-coset-graph}, which constructs a $\poly(\constraints)$-time distinguisher.
First, we compute the rate $R$ of the $\RM\left( \logConstraints, \frac{\logConstraints+d}{2} \right )$.
For convenience, let $r = \frac{\logConstraints + d}{2}$.

We rewrite the rate as the probability that a binomial random variable with parameters $(\logConstraints, 1/2)$ is at most $r$:
\[R = 2^{-\logConstraints} \cdot \sum_{i=0}^{r} \binom{\logConstraints}{i} = \Pr\left [ \Bin(\logConstraints, 1/2) \leq r \right ]. \]
Then, consider the normalized random variable $S \coloneqq \left( \Bin(\logConstraints, 1/2) - (\logConstraints/2) \right) / (\sqrt{\logConstraints}/ 2)$.
The equivalent event of interest in terms of $S$ is when $S \leq 2\cdot  (r - (\logConstraints/2))/\sqrt{\logConstraints} =  c/2$, that is,
\[\Pr[\Bin(\logConstraints, 1/2) \leq r] = \Pr[S \leq c/2],\]
where we recall that $c$ is the constant such that $d = c \sqrt{\logConstraints}$. 
Since the absolute centered third moment of the $\Ber(1/2)$ random variable is $1/8$, the Berry-Ess\'een Theorem (see Theorem~\ref{thm:berry-esseen}) gives
\[\Pr[S \leq c/2 ] = \Pr \left[ Z \leq c/2 \right ] + O \left ( \frac{1}{\sqrt{\logConstraints}}\right ), \]
for $Z \sim N(0, 1)$.
Any $c$ satisfying Equation~\eqref{eqn:hypo-1}
implies that $1 - R \geq 0.4  + O \left (\frac{1}{\sqrt{\logConstraints}} \right)$ so that $\eta  = 1/3 \leq 1 - R + O(1/\log \logConstraints)$.

For the other hypothesis, we need to show that $\eta \cdot 2^\logConstraints \leq 2^{\logConstraints \cdot H_2((1-(d/\logConstraints))/2) }$ for sufficiently large $\logConstraints$.
Observe that $d/\logConstraints = c/\sqrt{\logConstraints}$.
By the Taylor expansion of $H_2$ around $1/2$, we have that
\[\logConstraints \cdot H_2((1 - (d/\logConstraints))/2) = \logConstraints \cdot \left ( 1 - \frac{2}{\ln 2} \cdot \left (\frac{c}{\sqrt{\logConstraints}} \right )^2 + O(\logConstraints^{-2})  \right) = \logConstraints - \frac{2 c^2}{ \ln 2} + O(\logConstraints^{-1}).\]
Therefore, as long as $(1/3) \leq 2^{- 2 c^2 / (\ln 2)}$, any $c$ satisfying Equation~\eqref{eqn:hypo-2} we have
\[ (1/3) \cdot \constraints \lesssim 2^{\logConstraints \cdot H_2((1-(d/\logConstraints))/2)}.\]

Since both hypotheses are satisfied, Theorem~\ref{thm:decoding-coset-graph} implies that there exists an algorithm with running time $\poly(\constraints)$ that solves the $\eta$-noisy $k$-XOR distinguishing problem with constant noise rate $\eta = 1/3$.
\end{proof}

\subsection{The Setting of Polynomially Many Constraints}

Assuming Reed-Muller codes efficiently achieve capacity of the BEC allows us to obtain a stronger theorem statement in which the number of constraints is polynomially related to the number of variables.
As mentioned before, a weaker conjecture (Conjecture~\ref{conj:RM-random-erasure-weak}) leads to the same conclusion.

\begin{theorem}[Main Theorem~\ref{thm:main2}]
\label{thm:polynomial-size-weak}
Assume Conjecture~\ref{conj:RM-random-erasure-weak}.
Then for every $\alpha, c >0$ there exists an infinite family of
$k$-left-regular constraint graphs $G=([\constraints]\cup[\variables],E)$ such that
\[
\constraints=\variables^{\Theta(1)},
\qquad
k=(\log \variables)^{O(1/\alpha)},
\]
the graph is $(\variables^{1-\alpha},1-o(1))$-expanding, and there is an algorithm running in time
$\poly(\constraints)$ that solves the $\eta$-noisy $k$-XOR distinguishing problem on these graphs for $\eta = \variables^{-c}$.
\end{theorem}
\begin{proof}

Fix constants $\alpha, c > 0$. Let $\xi > 0, \zeta \geq 1$ be constants given by Conjecture~\ref{conj:RM-random-erasure-weak}.
    Let\[
\lambda(\beta)\coloneqq 1-H_2\!\left(\frac{1-\beta}{2}\right)
\]
for a value of $\beta$ we will choose shortly.
By the Taylor expansion of $H_2$ around $1/2$, as $\beta \to 0$ we have 
\[
\lambda(\beta)=\frac{\beta^2}{2\ln 2}+O(\beta^4).
\]
Let $\gamma_0 \coloneqq \frac{c \cdot \beta}{2}$, and choose a sufficiently small constant $\beta\in(0,1)$ such that
\begin{equation}
\lambda(\beta)< \min\left(\xi,\; \frac{c \cdot \beta}{4 \zeta}  \right) < \gamma_0.
\label{eq:beta-choice-for-weak-thm}
\end{equation}
By Corollary~\ref{cor:expanding-coset-family}, there is an infinite family of
explicit $(\logConstraints,k,d)$-coset graphs
\[
G_\logConstraints=([\constraints]\cup[\variables],E_\logConstraints)
\]
which is $(\variables^{1-\alpha},1-o(1))$-expanding with
\begin{itemize}
    \item $\constraints=2^\logConstraints$,
    \item $d=\beta \logConstraints+\Theta(\log \logConstraints)$,
    \item $k=(\logConstraints/\log \logConstraints)^{\Theta(1/\alpha)}$.
\end{itemize}
We claim that $\constraints = \variables^{\Theta(1)}$ and that $k = (\log \variables)^{\Theta(1/\alpha)}$.
Since $G_\logConstraints$ is an $(\logConstraints,k,d)$-coset graph, its right side has size $\variables=k\cdot 2^d$.
Because $d=\beta \logConstraints+\Theta(\log \logConstraints)$ and $\log k = O(\log \logConstraints)$, we have
\[\log \variables = d+\log k = \beta \logConstraints + \Theta(\log \logConstraints) = (\beta + o(1)) \cdot \logConstraints,\]
implying immediately that
\begin{equation}
    \label{eqn:M-asymptotic}
    \variables = \constraints^{\beta + o(1)}.
\end{equation}
and implying immediately that
$\logConstraints=\Theta(\log \variables)$ so that 
\[\boxed{\constraints=2^\logConstraints = \variables^{\Theta(1)}.}\]
Also, since $\logConstraints=\Theta (\log  \variables)$, the left degree satisfies
\[ 
\boxed{k=(\logConstraints/\log \logConstraints)^{\Theta(1/\alpha)}=(\log \variables)^{\Theta(1/\alpha)}.}
\]

We now verify the hypotheses of Theorem~\ref{thm:decoding-coset-graph-weak}.
First observe that,
\[
p\coloneqq d/\logConstraints = \beta + \Theta \left(\frac{\log \logConstraints}{\logConstraints} \right) = \beta + o(1).
\]
For convenience let $r\coloneqq \frac{\logConstraints+d}{2}$.
Using the continuity of the binary entropy function
\begin{equation}
\label{eqn:continuity-for-weak-thm}
H_2 \left(\frac{1-p}{2}\right)
=H_2 \left(\frac{1-\beta}{2}\right)+o(1)
=1-\lambda(\beta)+o(1).
\end{equation}
For convenience let $r\coloneqq \frac{\logConstraints+d}{2}$.
We now show the first hypothesis holds for Theorem~\ref{thm:decoding-coset-graph-weak}, namely that for sufficiently large $\logConstraints$, 
\[ \constraints^{-\xi} < \varepsilon_{\logConstraints, d} < 1/2.\]
Observe that
\[
\varepsilon_{\logConstraints,d}
\coloneqq 1-\Rate\!\left(\RM\!\left(\logConstraints,\frac{\logConstraints+d}{2}\right)\right)
=2^{-\logConstraints}\sum_{i<\frac{\logConstraints-d}{2}}\binom{\logConstraints}{i}.
\]
By applying the standard Hamming ball asymptotics and then applying Equation~\eqref{eqn:continuity-for-weak-thm}, we have
\begin{equation}
\label{eqn:vareps}
   \varepsilon_{\logConstraints,d}
=2^{-\logConstraints(1-H_2((1-p)/2))+o(\logConstraints)}
= \constraints^{-\lambda(\beta)+o(1)}. 
\end{equation}
This quantity is less than $1/2$ for all sufficiently large $\logConstraints$ since $\constraints = 2^\logConstraints$.
Then recall that $\beta$ was chosen so that $\lambda(\beta) < \xi$. 
Therefore, for sufficiently large $\logConstraints$, we have the desired property 
\[\boxed{\constraints^{-\xi} < \varepsilon_{\logConstraints, d} < 1/2.}\]
Now we show that $\eta = \variables^{-c}$ satisfies
$\eta < \varepsilon_{\logConstraints, d}^{\zeta}$.
Equation~\eqref{eqn:M-asymptotic} implies $\eta = \constraints^{- c \cdot \beta + o(1)}$.
Equation~\eqref{eqn:vareps} implies
\[
\varepsilon_{\logConstraints,d}^{\zeta}
= \constraints^{-\zeta\cdot \lambda(\beta)+o(1)}.
\]
Equation~\eqref{eq:beta-choice-for-weak-thm} implies $\frac{\lambda(\beta)}{\beta} < \frac{c}{\zeta}$, which implies the desired statement that for sufficiently large $\logConstraints$, 
\[
\boxed{\eta < \varepsilon_{\logConstraints, d}^{\zeta}.}
\]
Finally we show that $ \eta \cdot \constraints \lesssim \constraints^{H_2((1-p)/2)}$.
Since $\eta = \variables^{-c}$, Equation~\eqref{eqn:M-asymptotic} implies that
\[\eta \cdot \constraints = \constraints^{1 - c \cdot \beta + o(1)}.\]
Equation~\eqref{eqn:continuity-for-weak-thm} implies that
\[\constraints^{H_2((1-p)/2)} = \constraints^{1 - \lambda(\beta) + o(1)}.\]
Equation~\eqref{eq:beta-choice-for-weak-thm} implies that $1 - c\cdot \beta < 1 - \lambda(\beta)$ since $\lambda(\beta) < c \cdot \beta/4$ (since $\zeta \geq 1$).
Therefore we conclude that 
\[\boxed{\eta \cdot \constraints \lesssim \constraints^{H_2((1-p)/2)}.}\]
Since all hypotheses of Theorem~\ref{thm:decoding-coset-graph-weak} hold, there exists a $\poly(\constraints)$ time algorithm that solves the $\eta$-noisy XOR distinguishing problem for $G_\logConstraints$ with distinguishing advantage $1 - o(1)$. 
\end{proof}
\section*{Acknowledgements}
We thank Andrej Bogdanov for helpful discussions. PL and AR are supported by European Research Council (ERC) under the EU’s Horizon 2020 research and innovation programme (Grant agreement No. 101019547). PL is additionally supported by Stellar Foundation grant and AR by Cariplo CRYPTONOMEX grant. MS is supported in part by a Simons Investigator Award and NSF Award CCF 2152413. Part of MS's work done while visiting Bocconi University.
\printbibliography

@inproceedings{DBLP:conf/crypto/BogdanovRT25,
  author       = {Andrej Bogdanov and
                  Alon Rosen and
                  Kel Zin Tan},
  title        = {Sample Efficient Search to Decision for kLIN},
  booktitle    = {Advances in Cryptology - {CRYPTO} 2025 - 45th Annual International
                  Cryptology Conference, Santa Barbara, CA, USA, August 17-21, 2025,
                  Proceedings, Part {I}},
  series       = {Lecture Notes in Computer Science},
  pages        = {203--220},
  publisher    = {Springer},
  year         = {2025}
}

@inproceedings{FOCS:PV05,
  author       = {Farzad Parvaresh and
                  Alexander Vardy},
  title        = {Correcting Errors Beyond the Guruswami-Sudan Radius in Polynomial
                  Time},
  booktitle    = {46th Annual {IEEE} Symposium on Foundations of Computer Science, {FOCS}
                  2005, Pittsburgh, PA, USA, October 23-25, 2005, Proceedings},
  pages        = {285--294},
  publisher    = {{IEEE} Computer Society},
  year         = {2005},
  url          = {https://doi.org/10.1109/SFCS.2005.29},
  doi          = {10.1109/SFCS.2005.29},
  bibsource    = {dblp computer science bibliography, https://dblp.org}
}

@article{DBLP:journals/siamcomp/Applebaum13,
  author       = {Benny Applebaum},
  title        = {Pseudorandom Generators with Long Stretch and Low Locality from Random
                  Local One-Way Functions},
  journal      = {{SIAM} J. Comput.},
  volume       = {42},
  number       = {5},
  pages        = {2008--2037},
  year         = {2013},
  url          = {https://doi.org/10.1137/120884857},
  doi          = {10.1137/120884857},
  bibsource    = {dblp computer science bibliography, https://dblp.org}
}

@article{DBLP:journals/siamcomp/ApplebaumL18,
  author       = {Benny Applebaum and
                  Shachar Lovett},
  title        = {Algebraic Attacks against Random Local Functions and Their Countermeasures},
  journal      = {{SIAM} J. Comput.},
  volume       = {47},
  number       = {1},
  pages        = {52--79},
  year         = {2018},
  url          = {https://doi.org/10.1137/16M1085942},
  doi          = {10.1137/16M1085942},
  bibsource    = {dblp computer science bibliography, https://dblp.org}
}

@article{DBLP:journals/cc/BogdanovQ12,
  author       = {Andrej Bogdanov and
                  Youming Qiao},
  title        = {On the security of Goldreich's one-way function},
  journal      = {Comput. Complex.},
  volume       = {21},
  number       = {1},
  pages        = {83--127},
  year         = {2012},
  url          = {https://doi.org/10.1007/s00037-011-0034-0},
  doi          = {10.1007/S00037-011-0034-0},
  bibsource    = {dblp computer science bibliography, https://dblp.org}
}

@inproceedings{bogdanov-sabin-vasudevan,
author = {Andrej Bogdanov and Manuel Sabin and Prashant Nalini Vasudevan},
title = {XOR Codes and Sparse Learning Parity with Noise},
booktitle = {Proceedings of 2019 ACM-SIAM Symposium on Discrete Algorithms (SODA)},
pages = {986-1004},
year = {2019},
doi = {10.1137/1.9781611975482.61},
URL = {https://epubs.siam.org/doi/abs/10.1137/1.9781611975482.61},
eprint = {https://epubs.siam.org/doi/pdf/10.1137/1.9781611975482.61}}

@article{DBLP:journals/joc/ApplebaumBR16,
  author       = {Benny Applebaum and
                  Andrej Bogdanov and
                  Alon Rosen},
  title        = {A Dichotomy for Local Small-Bias Generators},
  journal      = {J. Cryptol.},
  volume       = {29},
  number       = {3},
  pages        = {577--596},
  year         = {2016},
  url          = {https://doi.org/10.1007/s00145-015-9202-8},
  doi          = {10.1007/S00145-015-9202-8},
  bibsource    = {dblp computer science bibliography, https://dblp.org}
}

@article{DBLP:journals/rsa/MosselST06,
  author       = {Elchanan Mossel and
                  Amir Shpilka and
                  Luca Trevisan},
  title        = {On epsilon-biased generators in $\mathrm{NC}^0$},
  journal      = {Random Struct. Algorithms},
  volume       = {29},
  number       = {1},
  pages        = {56--81},
  year         = {2006},
  url          = {https://doi.org/10.1002/rsa.20112},
  doi          = {10.1002/RSA.20112},
  bibsource    = {dblp computer science bibliography, https://dblp.org}
}

@inproceedings{DBLP:conf/crypto/BlumFKL93,
  author       = {Avrim Blum and
                  Merrick L. Furst and
                  Michael J. Kearns and
                  Richard J. Lipton},
  title        = {Cryptographic Primitives Based on Hard Learning Problems},
  booktitle    = {13th Annual International Cryptology Conference, CRYPTO 1993},
  pages        = {278--291},
  year         = {1993},
  url          = {https://doi.org/10.1007/3-540-48329-2\_24},
  doi          = {10.1007/3-540-48329-2\_24},
  bibsource    = {dblp computer science bibliography, https://dblp.org}
}

@inproceedings{GHJS25,
author = {Ghosal, Riddhi and Hair, Isaac M. and Jain, Aayush and Sahai, Amit},
title = {Using the Planted Clique Conjecture for Cryptography: Public-Key Encryption from Planted Clique and Noisy k-LIN over Expanders},
year = {2025},
isbn = {9798400715105},
publisher = {Association for Computing Machinery},
address = {New York, NY, USA},
url = {https://doi.org/10.1145/3717823.3718306},
doi = {10.1145/3717823.3718306},
booktitle = {Proceedings of the 57th Annual ACM Symposium on Theory of Computing},
pages = {1921–1932},
numpages = {12},
location = {Prague, Czechia},
series = {STOC '25}
}

@article{ASW14,
  author       = {Emmanuel Abbe and
                  Amir Shpilka and
                  Avi Wigderson},
  title        = {Reed-Muller codes for random erasures and errors},
  journal      = {CoRR},
  volume       = {abs/1411.4590},
  year         = {2014},
  url          = {http://arxiv.org/abs/1411.4590},
  eprinttype    = {arXiv},
  eprint       = {1411.4590},
  bibsource    = {dblp computer science bibliography, https://dblp.org}
}

@misc{SS18,
      title={On the Performance of Reed-Muller Codes with respect to Random Errors and Erasures}, 
      author={Ori Sberlo and Amir Shpilka},
      year={2018},
      eprint={1811.12447},
      archivePrefix={arXiv},
      primaryClass={cs.IT},
      url={https://arxiv.org/abs/1811.12447}, 
}

@article{SSV15,
  author       = {Ramprasad Saptharishi and
                  Amir Shpilka and
                  Ben Lee Volk},
  title        = {Decoding high rate Reed-Muller codes from random errors in near linear
                  time},
  journal      = {CoRR},
  volume       = {abs/1503.09092},
  year         = {2015},
  url          = {http://arxiv.org/abs/1503.09092},
  eprinttype    = {arXiv},
  eprint       = {1503.09092},
  bibsource    = {dblp computer science bibliography, https://dblp.org}
}

@article{kkmpsu16,
  author       = {Shrinivas Kudekar and
                  Santhosh Kumar and
                  Marco Mondelli and
                  Henry D. Pfister and
                  Eren Sasoglu and
                  R{\"{u}}diger L. Urbanke},
  title        = {Reed-Muller Codes Achieve Capacity on Erasure Channels},
  journal      = {{IEEE} Trans. Inf. Theory},
  volume       = {63},
  number       = {7},
  pages        = {4298--4316},
  year         = {2017},
  url          = {https://doi.org/10.1109/TIT.2017.2673829},
  doi          = {10.1109/TIT.2017.2673829},
  bibsource    = {dblp computer science bibliography, https://dblp.org}
}

@inproceedings{STOC:KKMPSU16,
  author       = {Shrinivas Kudekar and
                  Santhosh Kumar and
                  Marco Mondelli and
                  Henry D. Pfister and
                  Eren Sasoglu and
                  R{\"{u}}diger L. Urbanke},
  editor       = {Daniel Wichs and
                  Yishay Mansour},
  title        = {Reed-Muller codes achieve capacity on erasure channels},
  booktitle    = {Proceedings of the 48th Annual {ACM} {SIGACT} Symposium on Theory
                  of Computing, {STOC} 2016, Cambridge, MA, USA, June 18-21, 2016},
  pages        = {658--669},
  publisher    = {{ACM}},
  year         = {2016},
  url          = {https://doi.org/10.1145/2897518.2897584},
  doi          = {10.1145/2897518.2897584},
  bibsource    = {dblp computer science bibliography, https://dblp.org}
}

@article{GUV09, author = {Guruswami, Venkatesan and Umans, Christopher and Vadhan, Salil}, title = {Unbalanced expanders and randomness extractors from Parvaresh--Vardy codes}, year = {2009}, issue_date = {June 2009}, publisher = {Association for Computing Machinery}, address = {New York, NY, USA}, volume = {56}, number = {4}, issn = {0004-5411}, url = {https://doi.org/10.1145/1538902.1538904}, doi = {10.1145/1538902.1538904}, journal = {J. ACM}, month = jul, articleno = {20}, numpages = {34} }

@article{vadhan2012,
  title={Pseudorandomness},
  author={Vadhan, Salil P},
  journal={Foundations and Trends{\textregistered} in Theoretical Computer Science},
  volume={7},
  number={1-3},
  pages={1--336},
  year={2012},
  publisher={Emerald Publishing Limited}
}

@inproceedings{FOCS:Alekhnovich03,
  author       = {Michael Alekhnovich},
  title        = {More on Average Case vs Approximation Complexity},
  booktitle    = {44th Symposium on Foundations of Computer Science, {FOCS} 2003, Cambridge,
                  MA, USA, October 11-14, 2003, Proceedings},
  pages        = {298--307},
  publisher    = {{IEEE} Computer Society},
  year         = {2003},
  url          = {https://doi.org/10.1109/SFCS.2003.1238204},
  doi          = {10.1109/SFCS.2003.1238204},
  bibsource    = {dblp computer science bibliography, https://dblp.org}
}

@article{arXiv:BHLM25,
  author       = {Arpon Basu and
                  Jun{-}Ting Hsieh and
                  Andrew D. Lin and
                  Peter Manohar},
  title        = {Solving Random Planted CSPs below the $n^{k/2}$ Threshold},
  journal      = {CoRR},
  volume       = {abs/2507.10833},
  year         = {2025},
  url          = {https://doi.org/10.48550/arXiv.2507.10833},
  doi          = {10.48550/ARXIV.2507.10833},
  eprinttype   = {arXiv},
  eprint       = {2507.10833},
  bibsource    = {dblp computer science bibliography, https://dblp.org}
}

@inproceedings{NIPS:FPV15,
  author       = {Vitaly Feldman and
                  Will Perkins and
                  Santosh S. Vempala},
  editor       = {Corinna Cortes and
                  Neil D. Lawrence and
                  Daniel D. Lee and
                  Masashi Sugiyama and
                  Roman Garnett},
  title        = {Subsampled Power Iteration: a Unified Algorithm for Block Models and
                  Planted CSP's},
  booktitle    = {Advances in Neural Information Processing Systems 28: Annual Conference
                  on Neural Information Processing Systems 2015, December 7-12, 2015,
                  Montreal, Quebec, Canada},
  pages        = {2836--2844},
  year         = {2015},
  url          = {https://proceedings.neurips.cc/paper/2015/hash/9597353e41e6957b5e7aa79214fcb256-Abstract.html},
  bibsource    = {dblp computer science bibliography, https://dblp.org}
}

@inproceedings{STOC:RRS17,
  author       = {Prasad Raghavendra and
                  Satish Rao and
                  Tselil Schramm},
  editor       = {Hamed Hatami and
                  Pierre McKenzie and
                  Valerie King},
  title        = {Strongly refuting random CSPs below the spectral threshold},
  booktitle    = {Proceedings of the 49th Annual {ACM} {SIGACT} Symposium on Theory
                  of Computing, {STOC} 2017, Montreal, QC, Canada, June 19-23, 2017},
  pages        = {121--131},
  publisher    = {{ACM}},
  year         = {2017},
  url          = {https://doi.org/10.1145/3055399.3055417},
  doi          = {10.1145/3055399.3055417},
  bibsource    = {dblp computer science bibliography, https://dblp.org}
}

@inproceedings{STOC:KMOW17,
  author       = {Pravesh K. Kothari and
                  Ryuhei Mori and
                  Ryan O'Donnell and
                  David Witmer},
  editor       = {Hamed Hatami and
                  Pierre McKenzie and
                  Valerie King},
  title        = {Sum of squares lower bounds for refuting any {CSP}},
  booktitle    = {Proceedings of the 49th Annual {ACM} {SIGACT} Symposium on Theory
                  of Computing, {STOC} 2017, Montreal, QC, Canada, June 19-23, 2017},
  pages        = {132--145},
  publisher    = {{ACM}},
  year         = {2017},
  url          = {https://doi.org/10.1145/3055399.3055485},
  doi          = {10.1145/3055399.3055485},
  bibsource    = {dblp computer science bibliography, https://dblp.org}
}

@article{CC:Grigoriev01,
  author       = {Dima Grigoriev},
  title        = {Complexity of Positivstellensatz proofs for the knapsack},
  journal      = {Comput. Complex.},
  volume       = {10},
  number       = {2},
  pages        = {139--154},
  year         = {2001},
  url          = {https://doi.org/10.1007/s00037-001-8192-0},
  doi          = {10.1007/S00037-001-8192-0},
  bibsource    = {dblp computer science bibliography, https://dblp.org}
}

@book{baraksteurer16,
 title={Proofs, beliefs, and algorithms through the lens of sum-of-squares},
  author       = {Boaz Barak and
                  David Steurer},
  year = {2016},
}

@article{barak2014sum,
  title={Sum of squares upper bounds, lower bounds, and open questions},
  author={Barak, Boaz},
  journal={Lecture notes},
  year={2014}
}

@article{COJA:OGHLAN_GOERDT_LANKA_2007, title={Strong Refutation Heuristics for Random k-SAT}, volume={16}, DOI={10.1017/S096354830600784X}, number={1}, journal={Combinatorics, Probability and Computing}, author={Coja-Oghlan, Amin and Goerdt, Andreas and Lanka, Andr\'e}, year={2007}, pages={5–28}}

@misc{barak2016noisytensorcompletionsumofsquares,
      title={Noisy Tensor Completion via the Sum-of-Squares Hierarchy}, 
      author={Boaz Barak and Ankur Moitra},
      year={2016},
      eprint={1501.06521},
      archivePrefix={arXiv},
      primaryClass={cs.LG},
      url={https://arxiv.org/abs/1501.06521}, 
}

@misc{allen2015refuterandomcsp,
      title={How to refute a random CSP}, 
      author={Sarah R. Allen and Ryan O'Donnell and David Witmer},
      year={2015},
      eprint={1505.04383},
      archivePrefix={arXiv},
      primaryClass={cs.CC},
      url={https://arxiv.org/abs/1505.04383}, 
}

@misc{guruswami2023algorithmscertificatesbooleancsp,
      title={Algorithms and Certificates for Boolean CSP Refutation: "Smoothed is no harder than Random"}, 
      author={Venkatesan Guruswami and Pravesh K. Kothari and Peter Manohar},
      year={2023},
      eprint={2109.04415},
      archivePrefix={arXiv},
      primaryClass={cs.CC},
      url={https://arxiv.org/abs/2109.04415}, 
}

@inproceedings{applebaum2010public,
  title={Public-key cryptography from different assumptions},
  author={Applebaum, Benny and Barak, Boaz and Wigderson, Avi},
  booktitle={Proceedings of the forty-second ACM symposium on Theory of computing},
  pages={171--180},
  year={2010}
}

@misc{barak2023hiddenprogressdeeplearning,
      title={Hidden Progress in Deep Learning: SGD Learns Parities Near the Computational Limit}, 
      author={Boaz Barak and Benjamin L. Edelman and Surbhi Goel and Sham Kakade and Eran Malach and Cyril Zhang},
      year={2023},
      eprint={2207.08799},
      archivePrefix={arXiv},
      primaryClass={cs.LG},
      url={https://arxiv.org/abs/2207.08799}, 
}

@misc{raoyoutube,
  author       = "Anup Rao",
  title        = "Reed-Muller codes achieve capacity on the erasure channel",
  howpublished = "YouTube",
  month        = 12,
  year         = 2021,
  note         = "",
  url         = "https://www.youtube.com/watch?v=V9HCmPPz110"
}

@misc{abbe2019reedmullercodespolarize,
      title={Reed-Muller codes polarize}, 
      author={Emmanuel Abbe and Min Ye},
      year={2019},
      eprint={1901.11533},
      archivePrefix={arXiv},
      primaryClass={cs.IT},
      url={https://arxiv.org/abs/1901.11533}, 
}

@misc{hassani2018optimalscalingreedmullercodes,
      title={Almost Optimal Scaling of Reed-Muller Codes on BEC and BSC Channels}, 
      author={Hamed Hassani and Shrinivas Kudekar and Or Ordentlich and Yury Polyanskiy and Rüdiger Urbanke},
      year={2018},
      eprint={1801.09481},
      archivePrefix={arXiv},
      primaryClass={cs.IT},
      url={https://arxiv.org/abs/1801.09481}, 
}

@book{tao2023topics,
  title={Topics in random matrix theory},
  author={Tao, Terence},
  volume={132},
  year={2023},
  publisher={American Mathematical Society}
}

@incollection{Goldreich11,
  author       = {Oded Goldreich},
  editor       = {Oded Goldreich},
  title        = {Candidate One-Way Functions Based on Expander Graphs},
  booktitle    = {Studies in Complexity and Cryptography. Miscellanea on the Interplay
                  between Randomness and Computation - In Collaboration with Lidor Avigad,
                  Mihir Bellare, Zvika Brakerski, Shafi Goldwasser, Shai Halevi, Tali
                  Kaufman, Leonid Levin, Noam Nisan, Dana Ron, Madhu Sudan, Luca Trevisan,
                  Salil Vadhan, Avi Wigderson, David Zuckerman},
  series       = {Lecture Notes in Computer Science},
  pages        = {76--87},
  publisher    = {Springer},
  year         = {2011},
  url          = {https://doi.org/10.1007/978-3-642-22670-0\_10},
  doi          = {10.1007/978-3-642-22670-0\_10},
  bibsource    = {dblp computer science bibliography, https://dblp.org}
}

@article{ECCC:Goldreich00,
  author       = {Oded Goldreich},
  title        = {Candidate One-Way Functions Based on Expander Graphs},
  journal      = {Electron. Colloquium Comput. Complex.},
  volume       = {{TR00}},
  eid          = {{TR00-090}},
  year         = {2000},
  url          = {https://eccc.weizmann.ac.il/eccc-reports/2000/TR00-090/index.html},
  eprinttype   = {ECCC},
  eprint       = {TR00-090},
  bibsource    = {dblp computer science bibliography, https://dblp.org}
}

@inproceedings{ITCS:OST19,
  author       = {Igor Carboni Oliveira and
                  Rahul Santhanam and
                  Roei Tell},
  editor       = {Avrim Blum},
  title        = {Expander-Based Cryptography Meets Natural Proofs},
  booktitle    = {10th Innovations in Theoretical Computer Science Conference, {ITCS}
                  2019, San Diego, California, USA, January 10-12, 2019},
  series       = {LIPIcs},
  pages        = {18:1--18:14},
  publisher    = {Schloss Dagstuhl - Leibniz-Zentrum f{\"{u}}r Informatik},
  year         = {2019},
  url          = {https://doi.org/10.4230/LIPIcs.ITCS.2019.18},
  doi          = {10.4230/LIPICS.ITCS.2019.18},
  bibsource    = {dblp computer science bibliography, https://dblp.org}
}

@inproceedings{TCC:BKR23,
  author       = {Andrej Bogdanov and
                  Pravesh K. Kothari and
                  Alon Rosen},
  editor       = {Guy N. Rothblum and
                  Hoeteck Wee},
  title        = {Public-Key Encryption, Local Pseudorandom Generators, and the Low-Degree
                  Method},
  booktitle    = {Theory of Cryptography - 21st International Conference, {TCC} 2023,
                  Taipei, Taiwan, November 29 - December 2, 2023, Proceedings, Part
                  {I}},
  series       = {Lecture Notes in Computer Science},
  pages        = {268--285},
  publisher    = {Springer},
  year         = {2023},
  url          = {https://doi.org/10.1007/978-3-031-48615-9\_10},
  doi          = {10.1007/978-3-031-48615-9\_10},
  bibsource    = {dblp computer science bibliography, https://dblp.org}
}

@book{guruswami2012essential,
  title={Essential Coding Theory},
  author={Guruswami, Venkatesan and Rudra, Atri and Sudan, Madhu},
  publisher ={Draft available online},
  url = {http://www.cse.buffalo.edu/atri/courses/coding-theory/book},
  year={2012}
}

@inproceedings{feige2002relations,
  title={Relations between average case complexity and approximation complexity},
  author={Feige, Uriel},
  booktitle={Proceedings of the thirty-fourth annual ACM symposium on Theory of computing},
  pages={534--543},
  year={2002}
}

@INPROCEEDINGS{fko06,
  author={Feige, Uriel and Kim, Jeong Han and Ofek, Eran},
  booktitle={2006 47th Annual IEEE Symposium on Foundations of Computer Science (FOCS'06)}, 
  title={Witnesses for non-satisfiability of dense random 3CNF formulas}, 
  year={2006},
  volume={},
  number={},
  pages={497-508},
  doi={10.1109/FOCS.2006.78}}

@article{JACM:Braverman08,
  title={Polylogarithmic independence fools AC 0 circuits},
  author={Braverman, Mark},
  journal={Journal of the ACM (JACM)},
  volume={57},
  number={5},
  pages={1--10},
  year={2008},
  publisher={ACM New York, NY, USA}
}

@article{Journal:ASY21,
  author       = {Emmanuel Abbe and
                  Amir Shpilka and
                  Min Ye},
  title        = {Reed-Muller Codes: Theory and Algorithms},
  journal      = {{IEEE} Trans. Inf. Theory},
  volume       = {67},
  number       = {6},
  pages        = {3251--3277},
  year         = {2021},
  url          = {https://doi.org/10.1109/TIT.2020.3004749},
  doi          = {10.1109/TIT.2020.3004749},
  bibsource    = {dblp computer science bibliography, https://dblp.org}
}

@inproceedings{buhai2025quasipolynomiallowdegreeconjecturefalse,
  author       = {Rares{-}Darius Buhai and
                  Jun{-}Ting Hsieh and
                  Aayush Jain and
                  Pravesh K. Kothari},
  title        = {The Quasi-Polynomial Low-Degree Conjecture is False},
  booktitle    = {66th {IEEE} Annual Symposium on Foundations of Computer Science, {FOCS}
                  2025, Sydney, Australia, December 14-17, 2025},
  pages        = {2577--2590},
  publisher    = {{IEEE}},
  year         = {2025},
  url          = {https://doi.org/10.1109/FOCS63196.2025.00134},
  doi          = {10.1109/FOCS63196.2025.00134},
  bibsource    = {dblp computer science bibliography, https://dblp.org}
}

@book{hopkins2018statistical,
  title={Statistical inference and the sum of squares method},
  author={Hopkins, Samuel},
  year={2018},
  publisher={Cornell University}
}

@ARTICLE{journal:guruswami-xia,
  author={Guruswami, Venkatesan and Xia, Patrick},
  journal={IEEE Transactions on Information Theory}, 
  title={Polar Codes: Speed of Polarization and Polynomial Gap to Capacity}, 
  year={2015},
  volume={61},
  number={1},
  pages={3-16},
  doi={10.1109/TIT.2014.2371819}}

@article{journal:hassani2014finite,
  title={Finite-length scaling for polar codes},
  author={Hassani, Seyed Hamed and Alishahi, Kasra and Urbanke, R{\"u}diger L},
  journal={IEEE Transactions on Information Theory},
  volume={60},
  number={10},
  pages={5875--5898},
  year={2014},
  publisher={IEEE}
}

@ARTICLE{journal:gry22,
  author={Guruswami, Venkatesan and Riazanov, Andrii and Ye, Min},
  journal={IEEE Transactions on Information Theory}, 
  title={Arıkan Meets Shannon: Polar Codes With Near-Optimal Convergence to Channel Capacity}, 
  year={2022},
  volume={68},
  number={5},
  pages={2877-2919},
  doi={10.1109/TIT.2022.3146786}}

@article{journal:mondelli2014polar,
  title={From polar to Reed-Muller codes: A technique to improve the finite-length performance},
  author={Mondelli, Marco and Hassani, S Hamed and Urbanke, R{\"u}diger L},
  journal={IEEE Transactions on Communications},
  volume={62},
  number={9},
  pages={3084--3091},
  year={2014},
  publisher={IEEE}
}

@inproceedings{ITCS:HW21,
  author       = {Justin Holmgren and
                  Alexander S. Wein},
  editor       = {James R. Lee},
  title        = {Counterexamples to the Low-Degree Conjecture},
  booktitle    = {12th Innovations in Theoretical Computer Science Conference, {ITCS}
                  2021, Virtual Conference, January 6-8, 2021},
  series       = {LIPIcs},
  pages        = {75:1--75:9},
  publisher    = {Schloss Dagstuhl - Leibniz-Zentrum f{\"{u}}r Informatik},
  year         = {2021},
  url          = {https://doi.org/10.4230/LIPIcs.ITCS.2021.75},
  doi          = {10.4230/LIPICS.ITCS.2021.75},
  bibsource    = {dblp computer science bibliography, https://dblp.org}
}

\end{document}